\documentclass[
    a4paper,
    USenglish,
    cleveref,
    autoref,
    thm-restate,
    numberwithinsect,
]{lipics-v2021}
\usepackage{craigstyle}

\usepackage{xspace}
\usepackage{xcolor}
\usepackage{amssymb}
\usepackage[T1]{fontenc}
\usepackage{graphicx}
\usepackage{ifthen}
\usepackage{hyperref}
\usepackage{mdframed}
\usepackage{stmaryrd}
\usepackage{thm-restate}
\usepackage{tikz}
\usepackage{tikz-cd}
\usepackage{verbatim}

\newcommand{\FO}{\textup{FO}\xspace}
\newcommand{\FOord}{\textup{FO(<)}\xspace}
\newcommand{\FOordinv}{\textup{FO(<)}\ensuremath{_\textup{inv}}\xspace}
\newcommand{\SO}{\textup{SO}\xspace}
\newcommand{\ESO}{\textup{$\exists$SO}\xspace}
\newcommand{\USO}{\textup{$\forall$SO}\xspace}
\newcommand{\GFO}{\textrm{\upshape GFO}\xspace}
\newcommand{\GNFO}{\textrm{\upshape GNFO}\xspace}
\newcommand{\FOtwo}{\textrm{\upshape FO$^2$}\xspace}

\newcommand{\ML}{\textup{ML}\xspace}
\newcommand{\UNFO}{\textup{UNFO}\xspace}
\newcommand{\Ctwo}{\textup{C}$_2$\xspace}
\newcommand{\FL}{\textup{FL}\xspace}
\newcommand{\FF}{\textup{FF}\xspace}
\newcommand{\AF}{\textup{AF}\xspace}
\newcommand{\GFOtwo}{\textrm{\upshape GFO}$_2$\xspace}
\newcommand{\GFF}{\textrm{\upshape GFF}\xspace}
\newcommand{\GFL}{\textrm{\upshape GFL}\xspace}
\newcommand{\MLD}{\textrm{\upshape ML}$_D$\xspace}

\newcommand{\PA}{\textup{PA}\xspace}
\newcommand{\PrA}{\textup{PrA}\xspace}

\newcommand{\arity}{\textrm{arity}}
\newcommand{\relsig}{\mathop{\textup{relsig}}}
\newcommand{\sig}{\mathop{\textup{sig}}}
\newcommand{\fvar}{\mathop{\textup{fvar}}}
\newcommand{\acc}{S}
\newcommand{\BindPatt}{\mathop{\textup{BindPatt}}}
\newcommand{\leaf}{\textup{leaf}}
\newcommand{\firstchild}{\textup{firstchild}}
\newcommand{\secondchild}{\textup{secondchild}}

\newmdenv[
  backgroundcolor=lightgray!50,
  linecolor=gray,
  linewidth=1pt,
  skipabove=10pt,
  skipbelow=10pt,
  innerleftmargin=8pt,
  innerrightmargin=8pt,
  innertopmargin=5pt,
  innerbottommargin=5pt
]{boxtext}
\newenvironment{excursus}{\begin{boxtext}\textbf{Excursus.}\xspace\xspace}{\end{boxtext}}

\definecolor{darkgreen}{rgb}{0.0, 0.4, 0.13}

\newtheorem{question}{Question}

\author{Balder ten Cate}{University of Amsterdam}{b.d.tencate@uva.nl}{https://orcid.org/0000-0002-2538-5846}{Supported by the European Union’s Horizon 2020 research and innovation programme (MSCA-101031081).}

\author{Jesse Comer}{University of Pennsylvania}{jacomer@seas.upenn.edu}{https://orcid.org/0009-0006-9734-3457}{Acks}

\begin{document}

\title{Interpolation in First-Order Logic}

\maketitle

\begin{abstract}
In this chapter we give a basic overview of known results regarding Craig interpolation for first-order logic as well as for fragments of first-order logic. Our aim is to provide an entry point into the literature on interpolation theorems for first-order logic and fragments of first-order logic, and their applications. In particular, we cover a range of known refinements of the Craig interpolation theorem,
we discuss several important applications of interpolation in logic and computer science, we review known results about interpolation for important syntactic fragments of first-order logic, and we discuss the problem of computing interpolants.
\end{abstract}

\tableofcontents

\section{Introduction}
\label{sec:introduction}

In this chapter we give a basic overview of known results regarding Craig interpolation for first-order logic and many of its well-known fragments. Our aim is to provide an easy entry point into the literature. We focus on providing a global picture, intuition through examples, and pointers to further chapters in this book as well as other related literature. Other
useful articles on this topic include  \cite{DBLP:journals/synthese/Feferman08}, which provides a survey of the history and refinements of the theorem, \cite{vanbenthem2008}, which outlines important perspectives on the theorem and its significance, and \cite{vaananen2008}, describing the significance of the theorem in the study of abstract logics.

In Section \ref{sec:fo}, we recall the Craig interpolation theorem for first-order logic (\FO), as well as its variants. In Section \ref{sec:app}, we briefly survey some applications of the theorem, with references to other chapters of this book for additional details. In Section \ref{sec:frag}, we describe what is known about the Craig interpolation property among fragments of \FO. Section \ref{sec:theories} discusses interpolation in the context of background theories. Finally, Section \ref{sec:techniques} discusses how the Craig interpolation property can be proven, and how proof theory permits the effective construction of interpolants.
\section{Interpolation Theorems for First-Order Logic}
\label{sec:fo}

In this section, we review several prominent interpolation theorems that have been established for \FO. We first briefly recall the syntax of first-order logic in order to fix our notation. A \emph{signature} is a collection of relation, constant, and function symbols, where each relation symbol $R$ and each function symbol $f$ has a specified arity (denoted $\arity(R)$ or $\arity(f)$). We assume first-order formulas are generated by the following grammar:
\[
\phi ::= R(t_1, \ldots, t_n) \mid t_1=t_2 \mid \top \mid \phi \land \psi \mid \neg \phi \mid \exists x \phi
\]
where $R$ is an $n$-ary relation symbol, and $t_1, \ldots, t_n$ are terms built up from variables and constant symbols using function symbols. We treat disjunction, implication, and universal quantification as shorthand notations. We write $\sig(\phi)$ to denote the set of relation symbols, constants, and function symbols occurring in $\phi$. We write $\fvar(\phi)$ to denote the set of free variables of $\phi$. A \emph{sentence} is a formula without free variables. 

All results stated in this section hold for first-order logic both with and without equality. In the case without equality, it is important for the interpolation theorems that $\top$ is explicitly included in the syntax. In the case with equality, $\top$ could be treated as a shorthand for $\forall x(x=x)$. We defer discussion of techniques used to prove the interpolation theorems to Section~\ref{sec:techniques}.

\subsection{Craig Interpolation}

We begin with the Craig interpolation theorem (1957~\cite{craig1957three, craig1957linear}).

\begin{definition}[Craig interpolant]
Let $\phi,\psi$ be \FO formulas such that $\models \phi\to\psi$.
A \emph{Craig interpolant} for $\models \phi\to\psi$
is an \FO-formula $\vartheta$ such that
    the following hold:
    \begin{enumerate}
        \item $\models\phi\to\vartheta$,
        \item $\models\vartheta\to\psi$,
        \item $\sig(\vartheta)\subseteq \sig(\phi)\cap \sig(\psi)$,
        \item $\fvar(\vartheta)\subseteq \fvar(\phi)\cap \fvar(\psi)$.
    \end{enumerate}
\end{definition}

\begin{theorem}[Craig interpolation~\cite{craig1957three, craig1957linear}]
    Every valid \FO implication has a Craig interpolant.
\end{theorem}

\begin{example}\label{ex:craig}
Consider the formulas $\phi = \exists x Cat(x)\land \forall x(Cat(x)\to (Big(x)\land Green(x)))$ and $\psi = \exists x(Big(x)\land (Cat(x)\lor Dog(x)))$. Then
\[
\models \phi\to\psi
\]
In this case $\sig(\phi)=\{Cat,Big,Green\}$ and $\sig(\psi)=\{Cat,Dog,Big\}$. Therefore, a Craig interpolant for this valid entailment may only contain the relation symbols $Cat$ and $Big$. Examples of Craig interpolants for this implication include:
\begin{itemize}
    \item $\vartheta_1 = \exists x(Big(x)\land Cat(x))$, and
    \item $\vartheta_2 = \exists x (Cat(x))\land \forall x(Cat(x)\to Big(x))$
\end{itemize}
This also shows that Craig interpolants are not unique in general.
\lipicsEnd\end{example}

\subsection{Variants of the Interpolation Theorem}

The \emph{Lyndon interpolation theorem} refines the Craig interpolation theorem by distinguishing positive occurrences and negative occurrences of relation symbols. We write $\relsig^+(\phi)$ and $\relsig^-(\phi)$ to denote the set of relation symbols occurring \emph{positively}, respectively, \emph{negatively}, in $\phi$. Here, an occurrence of a relation symbol is positive if it is in the scope of an even number of negation symbols, and negative if it is in the scope of an odd number of negation symbols; we refer these terms as the \emph{polarity} of the occurrence.

\begin{definition}[Lyndon interpolant]
A \emph{Lyndon interpolant} for $\models\phi\to\psi$ is a Craig interpolant $\vartheta$ such that, in addition, the following hold:
\begin{enumerate}
\item $\relsig^+(\vartheta)\subseteq \relsig^+(\phi)\cap \relsig^+(\psi)$,
\item $\relsig^-(\vartheta)\subseteq \relsig^-(\phi)\cap \relsig^-(\psi)$,   
\end{enumerate}
\end{definition}

\begin{theorem}[Lyndon interpolation~\cite{Lyndon1959}]
Every valid \FO implication has a Lyndon interpolant.
\end{theorem}

Thus, the Lyndon interpolation theorem refines the Craig interpolation theorem by interpolating jointly over positive and negative occurrences of relation symbols. It is worth noting that, for \emph{constants and
function symbols}, no distinction is made between 
positive and negative occurrences.

\begin{example}\label{ex:lyndon}
Let us revisit Example~\ref{ex:craig}. In this case, $\relsig^+(\phi)=\{Cat, Big, Green\}$, $\relsig^-(\phi)=\{Cat\}$, $\relsig^+(\psi)=\{Big,Cat,Dog\}$ and $\relsig^-(\psi)=\emptyset$. It follows that a Lyndon interpolant for $\models\phi\to\psi$ may only contain positive occurrences of $Cat$ and $Big$ and no negative occurrences. In Example~\ref{ex:craig}, we listed two Craig interpolants. The first of these is also a Lyndon interpolant, while the second is not.
\lipicsEnd\end{example}

\begin{example}
    This example illustrates why relation symbols are treated differently from constants and function symbols in the Lyndon interpolation theorem.
Let $\phi = P(c)$ and let $\psi = \neg\exists x(x=c\land \neg P(x))$.
Clearly $\models\phi\to\psi$ (indeed, $\phi$ and $\psi$ are logically equivalent). The constant symbol $c$ occurs only positively (i.e., under an even number of negations) in $\phi$ and occurs only negatively (i.e., under an odd number of negations) in $\psi$. It is easy to see that there does not exist a Craig interpolant for $\models\phi\to\psi$ that does not contain an occurrence of $c$. In other words, if we were to require that the interpolant contains only constants that occur in $\phi$ and $\psi$ with the same polarity, then such an interpolation does not in general exist for a valid \FO implication. The same example can be easily adjusted to use a function symbol instead of a constant.
\lipicsEnd\end{example}

Several further refinements were subsequently obtained. 
For instance, Feferman \cite{Feferman1968:lectures} and Stern~\cite{Stern1975} proved interpolation theorems in the setting of first-order formulas over a many-sorted signature, which also interpolates over sorts. We omit the details here since it would require introducing many-sorted signatures, and also because Feferman's and Stern's interpolation theorems are subsumed by an even more powerful interpolation theorem that can be stated without reference to many-sorted signatures -- namely, the Otto interpolation theorem. For a set of unary predicates $\mathcal{U}$, we say that an \FO-formula $\phi$ is \emph{$\mathcal{U}$-relativized} if every quantifier in $\phi$ is relativized, i.e., is of the form $\exists x(U(x)\land\cdots)$ or $\forall x(U(x)\to \cdots)$, for some $U \in \mathcal{U}$.

\begin{theorem}[Otto interpolation~\cite{Otto2000:interpolation}]
\label{thm:otto}
For every pair of \FO formulas $\phi,\psi$ such that $\models \phi\to\psi$, and for all sets $\mathcal{U}$ of unary predicates, if $\phi$ and $\psi$ are both $\mathcal{U}$-relativized, then there is a Lyndon interpolant $\vartheta$ that is also $\mathcal{U}$-relativized.
\end{theorem}

Observe that a formula is $\emptyset$-relativized if and only if it is quantifier-free. Therefore, as a consequence of Otto interpolation, every valid quantifier-free entailment has a quantifier-free Lyndon interpolant.      This also implies Lyndon interpolation for propositional logic, although it is a very inefficient route to obtain the latter.

\begin{example}
The following implication is valid:
\[
\models (\exists x (Px \land Rxx)\land \forall x(Px \to Qx))\to (\exists x(Qx\land\neg Tx)\lor \exists x(Tx\land Rxx).
\]
Both formulas are $\mathcal{U}$-relativized with $\mathcal{U}=\{P,Q,T\}$. By the Otto interpolation theorem, there is a $\mathcal{U}$-relativized interpolant. This interpolant may furthermore only contain positive occurrences of $Q$ and positive occurrences of $R$ (and no occurrences of $P$ or $T$ whatsoever). An example of such an interpolant is $\exists x(Qx\land Rxx)$.
\lipicsEnd\end{example}

\begin{example}\label{ex:otto-optimal}
This example shows that Theorem~\ref{thm:otto} cannot be strengthened to interpolate on the relativizers (i.e., to require that all relativizers in the interpolant must occur as relativizers in both $\phi$ and $\psi$). Consider the valid \FO implication
\[
\models\forall x(Px\to \neg Qx)\to \forall x(Qx\to\neg Px).
\] 
Both formulas here are $\mathcal{U}$-relativized sentences, for $\mathcal{U}=\{P,Q\}$. However, they do not have any relativizers in common. If all relativizers in some interpolant occurred as relativizers in both $\phi$ and $\psi$, then that interpolant would be an $\emptyset$-relativized sentence. As we have already observed, an \FO formula is $\emptyset$-relativized if and only if it is quantifier-free. Therefore, the interpolant would have to be a quantifier-free sentence. The only candidates, up to logical equivalence, are $\top$ and $\bot$, and it is easy to see that neither is an interpolant.
\lipicsEnd\end{example}

The last interpolation theorem we mention here is the \emph{access interpolation theorem}~\cite{Benedikt16:generating,benedikt2016book}, a further refinement of Otto interpolation. 
It is discussed in detail also in~\refchapter{chapter:databases}.

\begin{definition}
An \emph{access method} is a pair $(R,X)$ where $R$ is a $n$-ary relation symbol and $X\subseteq\{1, \ldots, n\}$. 
\end{definition}
Intuitively, an access method (also known in the logic programming literature as a \emph{mode}) $(R,X)$ signifies a way to access the relation $R$, namely, by providing specific values for the arguments in $X$ as input, and retrieving the set of all matching tuples. The access method $(R,\emptyset)$ does not require specifying any values as input and therefore allows us to retrieve the full content of the relation $R$. An access method $(R,X)$ with $X=\{i_1, \ldots, i_n\}$ a non-empty set, on the other hand, only allows us to retrieve, for a specific tuple of values $a_1, \ldots, a_n$, the set of tuples in $R$ that have specified value $a_k$ in position $i_k$ for $k\leq n$. In particular, the access method $(R,\{1, \ldots, \arity(R)\})$ only allows testing if a specific tuple belongs to the relation $R$. There is a natural partial order on access methods: $(R,X)\preceq (R,Y)$ if $X\subseteq Y$. Intuitively, in this case, the former access method is ``at least as powerful as'' (or, ``can be used to implement'') the latter access method. For any set $\acc$ of access methods, we denote by $\acc^\uparrow$ the upwards closure of $\acc$ w.r.t.~the partial order $\preceq$. 

\begin{example}
Consider the \FO sentence $\forall x(Px\to Qx)$. We can determine whether this sentence holds in a structure $A$ using only the access methods $(P,\emptyset)$ (which allows us to retrieve the entire content of $P$) and $(Q,\{1\})$ (which allows us to test for any specific value whether it belongs to $Q$). Namely, we first retrieve the entire set $P^A$ and then we test, for each $a\in P^A$, whether or not $a \in Q^A$. Similarly, we can determine whether the sentence holds in a structure $A$ using access methods $(P,\emptyset)$ and $(Q,\emptyset)$ by retrieving the entire sets $P^A$ and $Q^A$ and testing if the former is a subset of the latter. On the other hand, having only the access methods $(P,\{1\})$ and $(Q,\{1\})$ is not enough, since we are not given any information about the domain of $A$.
\lipicsEnd\end{example}

The above example suggests a notion of answerability using access methods.
We say that a formula is \emph{answerable for membership} (or simply, \emph{answerable}) using a set of access methods $\acc$ if 
 given a structure $A$ and a tuple $\textbf{a}=a_1, \ldots, a_n$, 
    we can effectively test whether $A\models\phi(\textbf{a})$ 
    by a procedure that only accesses $A$ using the access methods in $\acc$. We will refrain here from giving a formal definition, but one can be found in~\refchapter{chapter:databases} based on the notion of an \emph{RA-plan} (cf.~also~\cite[Section~3.1]{benedikt2016book}).
Note: there are \FO formulas that are not answerable using \emph{any} set of access methods.
Examples include $\forall x P(x)$ and $\exists x \neg P(x)$. Indeed,  an \FO \emph{sentence} is answerable by a set of 
access methods only if it is \emph{domain-independent} (i.e., its truth value in a structure depends only on the 
denotation of the relation symbols and not
on the underlying domain of the structure, or, in other words, 
it ``ignores domain elements that do not participate in any relation'').%
\footnote{This can be derived from Corollary~\ref{cor:access-determinacy} below, by choosing $S$ to be the set of all access methods.}

Whether or not an \FO formula is answerable using a given set of access methods has a lot to do with the syntactic shape of the formula in question. Given an \FO formula $\phi$, we can associate to $\phi$ a set of access methods as follows:
\[
\begin{array}{lll}
\BindPatt(\top) = \BindPatt(x=y) &=& \emptyset, \\
\BindPatt((R(x_1,\ldots, x_n)) &=& \{(R,\{1, \ldots, n\})\}, \\
\BindPatt(\phi\land\psi) &=& \BindPatt(\phi)\cup \BindPatt(\psi), \\
\BindPatt(\neg\phi) &=& \BindPatt(\phi), \\
\BindPatt(\exists \textbf{x}(R(y_1, \ldots, y_n)\land \phi) &=& \BindPatt(\phi)\cup\{(R,X)\} \\
&&\quad\text{ where $X=\{i\leq n\mid y_i\notin \textbf{x}\}$}, \\
\BindPatt(\phi) &=& \text{undefined for all $\phi$ not of the above forms}
\end{array}
\]

In the penultimate clause above, $\textbf{x}$ denotes a tuple of variables (i.e., we allow for quantifying multiple variables at once). Furthermore, the binding atom $R(y_1, \ldots, y_n)$ is not required
to contain all free variables of $\phi$
(i.e., it is not 
required to be a ``guard'' in the sense of the 
guarded fragment~\cite{Andreka1998:Modal}).

We will say that a \FO formula is \emph{$\acc$-bounded} (for a set $\acc$ of access methods) if $\BindPatt(\phi)$ is well-defined and $\BindPatt(\phi)\subseteq \acc^\uparrow$.%
\footnote{It is not difficult to see that a \FO formula $\phi$ with
$\BindPatt(\phi)\subseteq \acc^\uparrow$ can always
be rewritten into an equivalent formula
$\phi'$ with $\BindPatt(\phi')\subseteq \acc$. However, this rewriting may require the use of the equality symbol (for instance, if $\acc=\{(R,\emptyset)\}$, then it may be necessary to rewrite $R(x,y)$ to 
$\exists uv(R(u,v)\land u=x\land v=y)$).
To ensure that all theorems stated in this section hold also in the equality-free case, it is necessary to phrase the definition of \emph{$\acc$-bounded} in terms of $S^\uparrow$ instead of $S$.}
Every $S$-bounded FO formula is answerable, in the sense described above, using the access methods $S$.

\begin{example}
To verify the truth of $\phi=\forall xy(R(x,y)\to S(x,y))$ in a structure $A$ using access methods, we need full access to the relation $R^A$, but we only need limited access to $S^A$ which allows us to test membership of specific tuples. This is reflected by the fact that $\BindPatt(\phi) = \{(R,\emptyset), (S,\{1,2\})\}$.
\lipicsEnd\end{example}

For simplicity, the following results are stated only for relational 
signatures (i.e., without constant symbols and function symbols). However, 
the results also hold true (suitably adapted) for signatures that include
constant symbols.

\begin{theorem}[Access interpolation~\cite{benedikt2016book}]
\label{thm:access-interpolation}
Fix any relational signature and set $\acc$ of access methods. For every pair of $\acc$-bounded \FO formulas $\phi,\psi$ such that $\models \phi\to\psi$ there is an $\acc$-bounded Lyndon interpolant $\vartheta$.\footnote{The original statement in~\cite{benedikt2016book} is slightly stronger, we simplified it here.}
\end{theorem}

As a corollary, we obtained a semantic characterization of $\acc$-boundedness in the form of a preservation theorem. In order to state this 
characterization, fix a set $\acc$ of access methods. 
For a given structure $A$ and tuple $\textbf{a}$, 
let us denote by 
$AccPart_\acc(A,\textbf{a})$ the ``part of $A$ that is accessible from $\textbf{a}$ using $\acc$'',
i.e., the smallest set $D\subseteq Dom(A)$ that contains $\textbf{a}$ and is closed under the following rule:
\begin{center}
    If $(a_1, \ldots a_n)\in R^A$, $(R,X)\in \acc$ and $\{a_i\mid i\in X\}\subseteq D$, then 
    $\{a_1,\ldots,a_n\}\subseteq D$
\end{center}
A FO-formula $\phi(\textbf{x})$ is \emph{access-determined} by $\acc$ if for all structures $A$ and tuples $\textbf{a}$,
$A\models\phi(\textbf{a})$ iff $A|D\models\phi(\textbf{a})$, where $D=AccPart_\acc(A,\textbf{a})$ and where $A|D$ is the substructure of $A$ induced by $D$.
Access determinacy can be viewed as a refined form of \emph{invariance for generated submodels} (cf.~\cite{Areces2001hybrid}) parameterized by a set of access methods. 
Access-determinacy is also a necessary precondition for 
answerability using a given set of access methods. 

\begin{corollary} 
\label{cor:access-determinacy}
Let $\phi$ be any \FO formula over a relational signature
and let $\acc$ be any set of access methods.
\begin{enumerate}
    \item If $\phi$ is $\acc$-bounded
    then $\phi$ is access-determined by $\acc$. 
    \item If $\phi$ is access-determined by $\acc$ then $\phi$ is equivalent to an $\acc$-bounded \FO-formula.
\end{enumerate}
\end{corollary}

\begin{proof}[Proof sketch]
    The first item can be proved by a formula induction. The second item follows from Theorem~\ref{thm:access-interpolation} as follows:
    assume that $\phi(x_1, \ldots, x_m)$ is access-determined by $S$.     Let $D$ be a fresh unary predicate. We can construct an $S\cup\{(D,\emptyset)\}$-bounded FO-formula
    $\chi_D(x_1, \ldots, x_m)$ that, intuitively speaking, expresses that the denotation of $D$
    includes the accessible part of the structure
    (with respect to $S$ and $x_1, \ldots, x_m$).
    It follows by the definition of access-determinacy that $\models (\chi(D) \land \phi^D)\to (\chi(D')\to \phi^{D'})$
    where $\phi^D$ is the syntactic relativization of $\phi$ by $D$ (i.e., the formula obtained from $\phi$ by relativizing every quantifier to $D$).
    It can then be shown that an interpolant for this valid entailment, as produced by Theorem~\ref{thm:access-interpolation}, is an $S$-bounded FO-formula that is equivalent to $\phi$.
    A detailed version of this argument can be found in Section 6 of~\refchapter{chapter:databases}.
\end{proof}

The Otto interpolation theorem (for relational signatures) follows from the access interpolation theorem, because an \FO-formula $\phi$ is $\mathcal{U}$-relativized 
if and only if $\phi$ is $\acc$-bounded for $\acc = \{(U,\emptyset)\mid U\in \mathcal{U}\}\cup\{(R,\{1,\ldots, \arity(R)\})\mid R\in \sig(\phi)\}$.

One may wonder if the above interpolation theorem
could be further strengthened to
$$\BindPatt(\vartheta)\subseteq \BindPatt(\phi)\cap \BindPatt(\psi).$$
The answer is negative, as follows from Example~\ref{ex:otto-optimal}.

So far, we have discussed Craig interpolation and several of its refinements. A number of other closely related notions exist that can be viewed as 
weak forms of Craig interpolation, including \emph{$\Delta$-interpolation} and \emph{Beth definability}, discussed in Sections~\ref{sec:delta} and~\ref{sec:beth}, respectively.

\subsection{Negative Results}

We now present two negative results. The first pertains to \emph{finite model theory}, i.e., the study of first-order logic when restricted to finite structures.

\begin{theorem}[Failure of Craig interpolation in the finite \cite{Gurevich1984:towards}]
\label{thm:failure-in-the-finite}
There is an \FO implication (over a  relational signature) that is valid in the finite, i.e.,
\[
\models_{\text{fin}} \phi\to\psi,
\]
and for which there is no \FO formula $\vartheta$ with $sig(\vartheta)\subseteq sig(\phi)\cap sig(\psi)$ satisfying $\models_{\text{fin}}\phi\to\vartheta$ and $\models_{\text{fin}}\vartheta\to\psi$.
\end{theorem}

\begin{proof}[Proof (sketch)]
Let $\phi$ be the \FO sentence with $\sig(\phi)=\{R\}$ expressing that $R$ is an equivalence relation and every equivalence class of $R$ contains precisely two elements. Similarly, let $\psi$ be the \FO sentence with $\sig(\phi)=\{S\}$ expressing that $S$ is an equivalence relation and that every equivalence class of $S$ contains precisely two elements except for one equivalence class, which contains three elements. Every finite model of $\phi$ has a domain of even cardinality, while every finite model of $\psi$ has a domain of odd cardinality. Therefore, $\models_{\text{fin}}\phi\to\neg\psi$. A Craig interpolant must have an empty signature and must express (over finite structures) that the domain has even cardinality. It is well known (and can be shown using standard techniques such as Ehrenfeucht-Fra\"{i}ss\'{e} games), that no such formula exists.
\end{proof}

Example~\ref{ex:order-invariant} also shows that several important consequences of Craig interpolation fail in the finite as well. The second negative result is concerned with \emph{uniform interpolation}, a
topic discussed in detail in~\refchapter{chapter:uniform}.

\begin{definition}[Uniform interpolant]
An \FO-formula $\vartheta$ is a \emph{uniform interpolant} for an \FO-formula $\phi$ w.r.t.~a signature $\sigma\subseteq \sig(\phi)$ if $\models\phi\to\vartheta$, $\sig(\vartheta)\subseteq\sigma$, and for all \FO-formulas $\psi$ satisfying $\models\phi\to\psi$ and  $\sig(\phi)\cap\sig(\psi)\subseteq \sigma$, it holds that $\models\vartheta\to\psi$.
\end{definition}

In other words, a uniform interpolant for an \FO-formula $\phi$ is a formula $\vartheta$ which is a Craig interpolant for \emph{all} valid implications $\models\phi\to\psi$, where $\vartheta$ depends only on $\phi$ and the signature $\sigma$. In the literature, the above notion of ``uniform interpolant'' is more commonly referred to as \emph{uniform post-interpolant}. There is also a dual notion of \emph{uniform pre-interpolant}. For simplicity, we only consider uniform post-interpolants here, and we will call them simply ``uniform interpolant''. Propositional logic, intuitionistic logic, and certain modal logics admit uniform interpolants. Uniform interpolation fails, however, for \FO \cite{henkin1963extension}\footnote{An interesting fact of history is that, by William Craig's own account, as told during an event in Berkeley celebrating his 90th birthday, the Craig interpolation theorem was obtained as the result of a failed attempt to prove uniform interpolation for first-order logic (\cite{vanbenthem2008} and \emph{Ph. Kolaitis,  personal communication}). For more history of the Craig interpolation theorem, see also \cite{craig2008elimination, craig2008road}.}.

\begin{theorem}[Failure of uniform interpolation for \FO, folklore]
    There is an \FO formula $\phi$ and a signature
    $\sigma\subseteq\sig(\phi)$ for which there is no uniform interpolant.
\end{theorem}

\begin{proof}
Let $\phi$ be the \FO-sentence in one relation symbol $R$, expressing that $R$ is a total linear order without end-points. Let $\sigma$ furthermore be the empty signature. Finally, let $\psi_n$ (for $n>1$) be the \FO-sentence $\exists x_1\ldots x_n\bigwedge_{1\leq i<j\leq n}x_i\neq x_j$, expressing the existence of $n$ distinct objects. Clearly, $\models\phi\to\phi_n$ for all $n$. A uniform interpolant would have to be a \FO-sentence $\vartheta$ that is implied by $\phi$ and that implies $\phi_n$ for all $n>1$. It can be shown (for instance, using a compactness argument) that no such $\vartheta$ exists.
\end{proof}

\section{Applications of Interpolation for First-Order Logic}
\label{sec:app}

In this section, we review a few ways in which interpolation theorems for \FO can be applied.

\subsection{\texorpdfstring{$\exists\text{SO}\cap\forall\text{SO}=\FO$}{ESO\&ASO=FO}}
\label{sec:delta}

The Craig interpolation theorem implies
something interesting about the relationship between
first-order logic (\FO) and second-order logic (\SO). Recall that second-order logic extends
first-order logic with second-order variables (of arbitrary arity)
and second-order quantifiers. 
Every second-order formula over a signature $\tau$ can be written in 
prenex form as 
\[ Q_1 X_1\ldots Q_n X_n . \phi\]
where $Q_1, \ldots, Q_n\in\{\forall,\exists\}$, 
$X_1, \ldots, X_n$ are second-order variables,
and $\phi$ is a first-order formula over the signature $\tau\cup\{X_1, \ldots, X_n\}$.
We write $\ESO$ and (resp. $\USO$) to denote \emph{existential second-order logic} (resp. \emph{universal second-order logic}), that is, the fragment of \SO consisting of formulas in the above prenex form where $Q_1=\cdots=Q_n=\exists$ (resp. $Q_1=\cdots=Q_n=\forall$).

\begin{example}\label{ex:so}
    The following \ESO sentence defines
    the class of (possibly infinite)  directed graphs that are \emph{3-colorable}:
    \[ \exists X_1 \exists X_2\exists X_3 \big(\forall x(X_1(x)\lor X_2(x)\lor X_3(x))\land \forall xy(R(x,y)\to\bigwedge_{i \leq 3} (X_i(x)\to \neg X_i(y)))\big)
    \]
Similarly, the following \USO sentence defines 
    the class of (possibly infinite) \emph{strongly connected}
    directed graphs: 
    \[ \forall X \big(\exists y(X(y))\land \forall yz(X(y)\land Ryz\to X(z)) \to \forall y(X(y))\big)\]
As a final example, the following \USO sentence defines the class of \emph{finite} structures:
\[
    \forall R\big((\forall x\exists y R(x,y)\land \forall xyz(R(x,y)\land R(y,z)\to R(x,z)))\to \exists x R(x,x)\big).
\]
None of these classes is \FO-definable.
\lipicsEnd\end{example}

\begin{theorem}[$\Delta$-interpolation]
    Let $K$ be a class of structures that is
    definable by a $\ESO$-sentence
    as well as by a $\USO$-sentence.
    Then $K$ is already definable by an \FO-sentence.
\end{theorem}

\begin{proof}
    Let $\phi = \exists R_1\ldots R_n.\phi'$ be
    the $\ESO$-sentence defining $K$, and 
    let $\psi = \forall S_1\ldots S_m.\psi'$ be
    the $\USO$-sentence defining $K$. 
    We may assume without loss of generality
    that $\{R_1, \ldots, R_n\}\cap\{S_1, \ldots, S_m\}=\emptyset$.
    From the fact that $\phi$ and $\psi$ are
    equivalent, 
    it follows that
    $\models \phi'\to\psi'$.
    Let $\vartheta$ be any Craig interpolant for this valid \FO-implication. In particular,
    $\models\phi'\to\vartheta$ and $\models\vartheta\to\psi'$. Since 
    $R_1\ldots R_n$ and $S_1\ldots S_m$ do not occur in $\vartheta$, it follows that
    $\models\phi\to\vartheta$ and $\models\vartheta\to\psi$. Since
    $\phi$ and $\psi$ are equivalent, it follows
    that $\vartheta$ is equivalent to $\phi$ and $\psi$ as well. Hence, $\vartheta$ defines the class $K$.
\end{proof}

The same argument used here shows something stronger, namely:
\emph{whenever $K\subseteq K'$ with $K$ an \ESO-definable class of structures and $K'$ a
\USO-definable class of structures, then there is an \FO-definable class of structures $K''$ with
$K\subseteq K''\subseteq K'$.} This
stronger ``\FO-separability'' property is in fact,
in a precise sense, equivalent to Craig interpolation (see the discussion in \refchapter{chapter:modeltheory} on the \emph{separation property}).

The name ``$\Delta$-interpolation'' originates from a notation system commonly used for denoting logical quantifier prenex hierarchies. More precisely, $\Sigma^1_n$ and $\Pi^1_n$ are fragments of second-order logic defined as follows (for $n\geq 0$): $\Sigma^1_0=\Pi^1_0=\FO$; $\Sigma^1_{n+1}$ consists of all second-order formulas of the form $\exists R_1\ldots R_n\phi$ for $\phi\in\Pi^1_n$ while $\Pi^1_{n+1}$ consists of all second-order formulas of the form $\forall R_1\ldots R_n\phi$ for $\phi\in\Sigma^1_n$. A class of structures  is said to be in $\Sigma^1_n$, or in $\Pi^1_n$, if it can be defined by a $\Sigma^1_n$-sentence, respectively, by a $\Pi^1_n$-sentence. Finally, a class of structures is in $\Delta^1_n$ is both $\Sigma^1_n$ and $\Pi^1_n$ (where $\Delta$ stands for \emph{Durchschnitt}, the German term for intersection). Following this notational convention, $\ESO = \Sigma^1_1$ and $\USO = \Pi^1_1$, and the above $\Delta$-interpolation states that $\Delta^1_1=\FO$. $\Delta$-interpolation is also known as \emph{Souslin-Kleene interpolation}, owing to the similarity it bears to the Souslin-Kleene theorem (\cite[Chapter 20]{shoenfield2017recursion}), as well as to Post’s theorem (\cite[Chapter 14]{shoenfield2017recursion}), in recursion theory and descriptive set theory. For further reading on Souslin-Kleene interpolation, see the discussion in \refchapter{chapter:modeltheory}.

\subsection{Preservation Theorems}

Another important application of interpolation theorems is that they imply preservation theorems.
We already saw an example of this  (Corollary~\ref{cor:access-determinacy}), but we will discuss this in some more detail here.
A preservation theorem is a theorem that
characterizes a semantic property of \FO formulas in terms of a syntactic property.
We will illustrate this with a concrete example.
Let us say that an \FO formula $\phi$ is \emph{upward monotone} in a relation symbol $R$ if the following
holds for all structures $A$ and $B$: 
if $A$ and $B$ have the same domain and the 
same interpretation for every symbol other than $R$,
and $R^A\subseteq R^B$, and 
$A,g\models\phi$ for some variable assignment $g$, then also $B,g\models\phi$.
The following preservation theorem follows 
from the Lyndon interpolation theorem:

\begin{theorem}[Lyndon preservation~\cite{Lyndon1959preserved}]
For all \FO formula $\phi$ and relation symbols $R$, the following are equivalent:
\begin{enumerate}
    \item $\phi$ is upward monotone in $R$.
    \item $\phi$ is equivalent to an \FO formula $\phi'$ with $R\not\in \relsig^-(\phi)$.
\end{enumerate}
\end{theorem}

\begin{proof}
We only treat the interesting direction, namely from 1 to 2. Let $\phi$ be upward monotone in $R$. It follows that the following implication is valid:
\[ \models \phi \to (\forall \textbf{x}(R\textbf{x}\to R'\textbf{x})\to \phi[R/R'])\]
where $R'$ is a fresh relation symbol of the same arity as $R$, and $\phi[R/R']$ is the formula obtained  by replacing every occurrence of $R$ in $\phi$ by $R'$. Let $\vartheta$ be a Lyndon interpolant for this valid \FO implication. Since $R$ occurs only positively in the consequent of the implication (in the scope of two implications), it follows that $R\not\in\relsig^-(\vartheta)$. Furthermore, it is easy to see that $\vartheta$ is equivalent to $\phi$.
\end{proof}

The same general technique can be used to derive many other preservation theorems. For instance, the \emph{{\L}os-Tarski preservation theorem} and the \emph{homomorphism preservation theorem} (cf.~\cite{hodges1992model}) can both be derived from the Otto interpolation theorem. Similarly, Feferman's preservation theorems characterizing the \FO formulas that are \emph{persistent for outer extensions} and the \FO formulas that are \emph{invariant for outer extensions} (a.k.a.~the formulas \emph{preserved by}, respectively, \emph{invariant for generated substructures})~\cite{Feferman1968outer,Areces2001hybrid} follow from the access interpolation theorem.
\subsection{The Beth Definability Theorem}
\label{sec:beth}

We introduce the basic idea of the Beth definability theorem 
through an example.

\begin{example}
Consider the binary relation
``taller-than'' and the unary relation ``tallest''.
There are a number of reasonable axioms that one could write down, governing the relationship between these two concepts. These could include, for instance, axioms expressing that ``taller-than'' is a linear order, the axiom $\forall xy(\text{taller-than}(y,x)\to \neg\text{tallest}(x))$, or the axiom $\exists x(\text{tallest}(x))$. Let $\Sigma$ be the \FO theory consisting of these axioms. 
It can then be seen that 
$\text{tallest}$ is definable in terms of
$\text{taller-than}$, relative to $\Sigma$.
Indeed, \[\Sigma\models\forall x(\text{tallest}(x)\leftrightarrow \neg\exists y(\text{taller-than}(y,x)))\]
Using terminology that we will introduce below, 
we call the above an \emph{explicit definition
of $\text{tallest}$ in terms of $\text{taller-than}$ relative to $\Sigma$}.

What about the converse: is $\text{taller-than}$ explicitly definable in terms of $\text{tallest}$ relative to $\Sigma$? It seems not. We can give a formal argument as 
follows: consider the following pair of structures
\begin{itemize}
    \item $A$ is the structure with domain $\{a,b,c\}$ where $\text{taller-than}^A = \{(a,b), (b,c), (a,c)\}$ and $\text{tallest}^A=\{a\}$.
    \item $B$ is the structure with domain $\{a,b,c\}$ where $\text{taller-than}^B = \{(a,c), (c,b), (a,b)\}$ and $\text{tallest}^A=\{a\}$.
\end{itemize}
Both structures satisfy $\Sigma$. 
Furthermore, both structures have the same domain, and the
$\text{tallest}$ relation has the same denotation 
in both structures. If $\text{taller-than}$ were 
explicitly definable in terms of $\text{tallest}$, the
$\text{taller-than}$ relation would therefore
also have to have the same denotation in both structures. But this is not the case. Therefore,
we can conclude that $\text{taller-than}$ is \emph{not}
explicitly definable in terms of $\text{tallest}$,
relative to $\Sigma$. 
\lipicsEnd\end{example}

This technique for
showing non-definability (by providing a pair of structures that agree on their domain and on the  base relations but that disagree on the relation to be defined) is known as Padoa's method~\cite{padoa1901essai}. The Beth definability theorem~\cite{beth1953padoa}, intuitively, states that 
Padoa's method is a sound and complete method for proving non-definability: whenever a relation $R$ 
is \emph{not} definable in terms of relations
$S_1, \ldots, S_n$ relative to a first-order theory
$\Sigma$, then this can be shown using Padoa's method. Formally:

\begin{definition}
    Let $\Sigma$ be an \FO theory, let $R\in\sig(\Sigma)$ and let $\tau\subseteq \sig(\Sigma)$.
    \begin{enumerate}
    \item An \emph{explicit definition} of $R$ in terms of $\tau$ relative to $\Sigma$ is a 
    FO-formula $\phi(x_1, \ldots, x_n)$, with $n=arity(R)$, such that $\Sigma\models\forall x_1\ldots x_n(R(x_1, \ldots, x_n)\leftrightarrow \phi(x_1, \ldots, x_n))$.
    \item $\Sigma$ \emph{implicitly defines} $R$ in terms of $\tau$ if
    for all models $A$ and $B$ of $\Sigma$ having the same domain, $S^A=S^B$ for all $S\in\tau$ implies
    $R^A=R^B$.
    \end{enumerate}
\end{definition}

\begin{example}
    Consider the theory
    $\Sigma = \{\forall xy (V_n(x,y)\leftrightarrow \phi_n(x,y))\mid n>1\}$,
    where $\phi_n(x,y)$ expresses the existence of a directed $R$-path of length $n$ from $x$ to $y$. Clearly, $\Sigma\models V_4(x,y)\leftrightarrow \exists u \left( V_2(x,u)\land V_2(u,y) \right)$. In other words, $V_4$ is explicitly definable in terms of $V_2$ relative to $\Sigma$. On the other hand, 
    $V_4$ is not explicitly definable in terms of $V_3$. Following Padoa's method, this can be proved
    by showing that $\Sigma$ does not implicitly define $V_4$ in terms of $V_3$ (we omit the details). What is less obvious is that $V_5$ \emph{is} explicitly definable in terms of 
    $V_3$ and $V_4$. We leave it as an exercise to the reader to find the 
    explicit definition (cf.~\cite{Afrati2011} for a solution).
\lipicsEnd\end{example}

\begin{theorem}[Projective Beth definability]
        Let $\Sigma$ be an \FO theory,  $R\in\sig(\Sigma)$ and  $\tau\subseteq \sig(\Sigma)$. There is an explicit
        definition of $R$ in terms of $\tau$ relative to $\Sigma$ if and only if $\Sigma$  implicitly defines $R$ in terms of $\tau$.
\end{theorem}

\begin{proof}
We will only discuss the more interesting right-to-left direction. To simplify the exposition, we will restrict attention to relational signatures.
Assume
$\Sigma$  implicitly defines $R$ in terms of $\tau$.
For each $S\in\sig(\Sigma)$, let $S'$ be a fresh relation symbol of the same arity as $S$.
Let $\Sigma'$ be the theory obtained by replacing each relation symbol in $\Sigma$
 by its primed copy. From 
the fact that $\Sigma$ implicitly defines $R$ in terms of $\tau$, it follows that 
$$\Sigma \cup \Sigma' \models \left( \bigwedge_{S\in\tau} \forall\textbf{x} \left( S\textbf{x} \leftrightarrow S'\textbf{x} \right) \to \left( R\textbf{y}\to R'\textbf{y} \right) \right)$$
By compactness, there exist finite $\Sigma_0\subseteq \Sigma$ and $\Sigma_0' \subseteq \Sigma'$ such that
$$\Sigma_0 \cup \Sigma_0' \models \left( \bigwedge_{S\in\tau} \forall\textbf{x} \left( S\textbf{x} \leftrightarrow S'\textbf{x} \right) \right) \to \left( R\textbf{y}\to R'\textbf{y} \right)$$
Rearranging the formulas a bit, we obtain
\[\models \left( \left( \bigwedge\Sigma_0 \right) \land R\textbf{y} \right) \to \left( \left( \bigwedge\Sigma_0' \right) \land \left( \bigwedge_{S\in\tau} \forall\textbf{x} \left( S\textbf{x}\leftrightarrow S'\textbf{x} \right) \right) \to R'\textbf{y} \right) \]
 Let $\vartheta$ be any Craig interpolant. 
Observe that $\sig(\vartheta)\subseteq\tau$.
It is not difficult to show that 
\[\Sigma\models\forall\textbf{y}(R\textbf{y}\to \vartheta(\textbf{y}))\]
and
\[\Sigma'\cup\{\forall\textbf{x}(S\textbf{x}\leftrightarrow S'\textbf{x})\mid S\in\tau\}\models\forall\textbf{y}( \vartheta(\textbf{y})\to R'\textbf{y})\]
Since first-order validity is preserved when we uniformly substitute a relation symbol by another relation symbol, it follows (by replacing all primed relation symbols by their unprimed version) that
$\Sigma\models\forall\textbf{y}( \vartheta(\textbf{y})\to R\textbf{y})$,
and therefore
\[\Sigma\models\forall\textbf{y}(R\textbf{y}\leftrightarrow \vartheta(\textbf{y}))\]
In other words, $\vartheta$ is an 
explicit definition for $R$ in terms of $\tau$
relative to $\Sigma$.
\end{proof}

In the literature, this is known as the 
\emph{projective} Beth definability theorem.
The non-projective Beth definability theorem (as originally proved in~\cite{beth1953padoa}) is the 
special case where $\tau=\sig(\Sigma)\setminus \{R\}$.

\begin{example}[Order-invariant \FOord = \FO]
\label{ex:order-invariant}
    This example shows the projective Beth definability theorem in action by using it as a stepping stone towards a further application.
    Fix  a signature $\tau$. We say that an \FO sentence $\phi$
    over $\tau\cup\{<\}$ is \emph{order-invariant} if
    for all $\tau$-structures $A$ and for all linear orders $<_1,<_2$ over
    the domain of $A$, it holds that $(A,<_1)\models\phi$ iff $(A,<_2)\models\phi$. If 
    $\phi$ is an order-invariant \FO sentence over $\tau\cup\{<\}$ and $A$ is a $\tau$-structure,
    we may simply write $A\models\phi$ to denote that $(A,<)\models\phi$ for an arbitrary linear order $<$ (because the choice of linear does not matter). By \emph{order-invariant first-order logic} (notation: \FOordinv) over a signature $\tau$ we mean the collection of order-invariant \FO sentences over the signature $\tau \cup \{<\}$. \FOordinv can be viewed as a language for describing $\tau$-structures. From this perspective,
    \FO is a fragment of  \FOordinv, because, if a 
    \FO sentence does not contain any occurrences of $<$, it is trivially order-invariant. It is conceivable, however, that 
    \FOordinv would be strictly more expressive than \FO.
    Indeed, this turns out to be the case \emph{in the finite}: there are 
    \FOord sentences that are order-invariant over finite structures, but that are not equivalent, over finite structures, to an \FO sentence in the signature without the order (cf.~\cite{Schweikardt2013:order})
    Over arbitrary (finite and infinite) structures, however, $\FOordinv$ has the same expressive power as $\FO$, and the
    projective Beth definability theorem can be used to prove this, as we will explain in a moment.
    
\begin{excursus}
A beautiful example of the expressive power of $\FOordinv$ on finite structures is given in~\cite{Potthoff94} (cf.~also~\cite{Schweikardt2013:order}): consider  \emph{finite binary trees} as $\{E, D\}$-structures where $E$ is the edge relation and $D$ is the descendant relation, i.e., the transitive closure of $E$. Let us furthermore restrict to \emph{perfect} finite binary trees, meaning that every non-leaf node has exactly two children, and that all leaves have equal distance from the root.

There is an order-invariant $\FO(E, D, <)$-sentence $\phi$ such that for every such perfect finite binary tree $T$, we have that $T \models\phi$ iff every leaf of $T$ has an even distance from the root. It can be shown using standard techniques (e.g.,  Ehrenfeucht-Fra\"{i}sse games) that there is no $\FO(E, D)$-sentence $\psi$ with the same property. We first define the shorthand notations
\begin{align*}
\leaf(z) &:= \neg\exists y E(x,y), ~\text{and} \\
\firstchild(x) &:= \neg\exists yz(E(y,x)\land E(y,z)\land z<x), ~\text{and} \\
\secondchild(x) &:= \exists yz(E(y,x)\land E(y,z)\land z<x).
\end{align*}
The formulas defining $\firstchild(x)$ and $\secondchild(x)$ use the ordering on the structure to label each non-root node as the first or second child of its parent. The core idea behind the construction is expressed by the formula
\begin{align*}
\exists z ( \leaf\left(z\right) &\land \firstchild\left(z\right) \land \\
&\forall xy\left(E\left(x,y\right)\land D\left(y,z\right)\to \left(\firstchild\left(y\right)\leftrightarrow \neg\firstchild\left(x\right)\right)\right)),
\end{align*}
which expresses the existence of a leaf node $n$, such that the path from the root to
$n$ is of the form as described by the regular expression ``$(\firstchild\cdot\secondchild)^*\cdot \firstchild$''. It follows that the above formula is true in perfect trees of even depth and false in perfect trees of odd depth. This is not a complete argument for $\FO\subsetneq\FOordinv$: more work is
needed to ensure that the formula is order-invariant over all structures (not only over perfect binary trees). Nevertheless, we hope that
the above explanation gives the basic intuition.
\end{excursus}

    As we will now show, it follows from the projective Beth definability theorem that, over arbitrary (i.e., finite and infinite) structures,
    every \FOordinv sentence is equivalent to an \FO sentence.
    Indeed, fix an arbitrary \FO sentence $\varphi$ over the signature $\tau \cup \{<\}$, and let $\Sigma$ be the \FO theory over the signature $\tau\cup\{<,Z\}$ expressing that $<$ is a linear order and that $Z\leftrightarrow\phi$, where $Z$ is a $0$-ary relation symbol not occurring in $\tau$. Then $\phi$ is order-invariant if and only if all structures with the same domain which agree on relations in $\tau$ (but not necessarily the order $<$) also agree on $Z$. In other words, $\phi$ is order invariant precisely when $\Sigma$ implicitly defines $Z$ in terms of $\tau$.
    When this holds, by the projective Beth definability theorem, there is an
    explicit definition of $Z$ in terms of $\tau$, i.e., $\Sigma\models Z\leftrightarrow \psi$ for some FO sentence $\psi$ over the original signature $\tau$ (i.e., without the linear order). It follows that $\phi$ and $\psi$ are equivalent in the sense that 
    both define the same class of $\tau$-structures.
    \lipicsEnd\end{example}

    There is much more to say about Beth definability. In particular,
    \refchapter{chapter:philosophy} 
    discusses relevant connections to various topics in philosophy. 
    
\subsection{Combining Knowledge Bases}

In the field of \emph{knowledge representation}, one important set of problems that have been studied in depth pertains to 
\emph{knowledge-base management}. This set of problems includes
\emph{combining} knowledge bases that have partially overlapping domains, as well as, for instance, \emph{restricting} a knowledge base to a  subdomain. Many of these tasks turn out to be intimately related to variants of Craig interpolation or uniform interpolation. See~\refchapter{chapter:kr} for an in-depth exposition of these connections.
Here, we will merely illustrate this through one basic result. 
In general, a knowledge base can be thought of as a theory specified in a suitable fragment of \FO logic. 
For simplicity, let us consider, here, a knowledge base to simply be an \FO theory. Then the following
theorem, which follows immediately from the Craig interpolation theorem, tells us that two knowledge bases can be safely combined with each other, as long as they do not contradict each other on statements in their common signature.

\begin{theorem}[Robinson joint consistency~\cite{Robinson1960}]
For all \FO theories $\Sigma_1$ and $\Sigma_2$, the following are equivalent:
\begin{enumerate}
    \item $\Sigma_1\cup\Sigma_2$ is consistent,
    \item There is no \FO sentence $\phi$ with
$\sig(\phi)\subseteq\sig(\Sigma_1)\cap\sig(\Sigma_2)$,
such 
that $\Sigma_1\models\phi$ and $\Sigma_2\models\neg\phi$.
\end{enumerate}
\end{theorem}

\begin{proof}
    The direction from 1 to 2 is immediate. 
    For the converse direction, 
    if $\Sigma_1\cup\Sigma_2$ is inconsistent,
    then, by compactness, there are finite
    subsets $\Sigma'_1\subseteq\Sigma_1$ and $\Sigma'_2\subseteq\Sigma_2$ such that
    $\Sigma'_1\cup\Sigma'_2$ is inconsistent.
    It follows that
    \[\models (\bigwedge \Sigma'_1)\to\neg(\bigwedge\Sigma'_2). \]
    and for any Craig interpolant $\varphi$ for this valid \FO implication, we have that $\Sigma_1\models\phi$ and $\Sigma_2\models\neg\phi$.
\end{proof}

\subsection{Formal Verification}

As will be discussed in depth in~\refchapter{chapter:verification}, variants of Craig interpolation have been fruitfully applied in the area of \emph{formal verification}. We illustrate this use case here by means of a simple example.

\begin{example}
\label{ex:verification}
Consider the following computer program.
\begin{quote}\tt
function f(int x): int \{ \\
\smash{~} // Input: an integer that is assumed to be divisible by 3 \\
\smash{~} // Output: an integer that is divisible by 3. \\
\\
\smash{~}  // assert: x is divisible by 3 \\
\smash{~}  x := x+1; \\
\smash{~}  // assert: \fbox{?} \\
\smash{~}  x := x+2; \\
\smash{~}  // assert: x is divisible by 3 \\
\smash{~}  return x; \\
\}
\end{quote}
If the initial assertion holds before the program execution then the final assertion holds afterwards. In order to prove this, we would insert an intermediate assertion in the spot marked by \fbox{?} and show that each instruction preserves the truth of the surrounding assertions (that is, the pre-condition of each instruction implies that the post-conditions holds after executing the instruction). What would be the right assertion to put in place of the \fbox{?}? The answer is  obvious, namely: \emph{$x$ is congruent to 1 modulo 3}.
However, one would like to have a method that allows us to find this answer automatically. 
As it turns out, Craig interpolation provides such a method.

Clearly, the reasoning here requires some arithmetic. We can think of it as reasoning relative to a first-order theory such as Peano arithmetic. In fact, for this example it suffices to use Presburger arithmetic (\PrA), i.e., the
first-order theory of the integers with addition and without multiplication. Note that 
\emph{x is divisible by 3} can be expressed in the language of \PrA 
as $\exists u(x=u+u+u)$ (which we will abbreviate, in what follows, as $\exists u(x=3u)$.
Next, in order to make the reasoning more precise, let us introduce 
variables $x_0, x_1, x_2$ to distinguish the values of $x$ at each stage of the computation. We can then formalize the correctness claim that we wish to prove as follows:
$$\PrA \models \left(\exists u\left(x_0 = 3u\right) \land x_1=x_0+1 \land x_2=x_1+2\right) \to \exists v \left(x_2=3v\right)$$
Rearranging the formula a bit, we get:
\begin{align*}
\models &\left(\left(\bigwedge \PrA\right)\land \left(\exists u\left(x_0 = 3u\right) \land x_1=x_0+1\right)\right) \to \left(x_2=x_1+2 \to \exists v \left(x_2=3v\right)\right)
\end{align*}
We can now see that the only shared free variable between the antecedent and the consequent is $x_1$. Therefore, there
should exist some interpolant with $x_1$ as its only free variable. Indeed, one such
interpolant is
$\exists u (x_1=3u+1)$.
It follows that
we can put $\exists u(x =3u+1)$ for the \fbox{?}. In fact,
we can even simplify this assertion a bit further, using the fact
that Presburger arithmetic admits quantifier elimination. This finally leads us to the solution we wanted, namely the assertion that $x$ is congruent to 1 modulo 3.
See also Section~\ref{sec:theories} where we consider interpolation 
relative to a FO theory in more detail. 
\lipicsEnd\end{example}

Similar techniques can be fruitfully applied for the formal verification of more complicated code involving loops or recursion (see~\refchapter{chapter:verification}).
\section{Interpolation for Decidable Fragments of First-Order Logic}
\label{sec:frag}

In this section, we discuss the situation regarding Craig interpolation for a number of important decidable fragment of first-order logic. See~\refchapter{chapter:modeltheory} for a similar overview of results about interpolation for extensions of first-order logic. This section is largely based on~\cite{FOSSACS24}. 

\subsection{Decidable Fragments}

The study of decidable fragments of \FO is a topic with a long history, dating back to the early 1900s (\cite{Lowenheim1915,Skolem1920}, cf.~also~\cite{Borger1997:classic}), and more actively pursued since the 1990s. Inspired by Vardi~\cite{Vardi1996:why}, who asked ``what makes modal logic so robustly decidable?'' and Andreka et al.~\cite{Andreka1998:Modal}, who asked ``what makes modal logic tick?'' many decidable fragments have been introduced and studied over the last 25 years that take inspiration from modal logic (\ML), which itself can be viewed as a fragment of \FO that features a restricted form of quantification. These include the following fragments, each of which naturally generalizes modal logic in a different way: the \emph{two-variable fragment} (\FOtwo)~\cite{Mortimer1975:languages}, the \emph{guarded fragment} (\GFO)~\cite{Andreka1998:Modal}, and the \emph{unary negation fragment} (\UNFO)~\cite{tencate2013:unary}. Further decidable extensions of these fragments were subsequently identified, including the \emph{two-variable fragment with counting quantifiers} (\Ctwo)~\cite{Graedel97:two} and the \emph{guarded negation fragment} (\GNFO)~\cite{Barany2015:guarded}. The latter can be viewed as a common generalization of \GFO and \UNFO. Many decidable logics used in computer science and AI, including various description logics and rule-based languages, can be translated into \GNFO and/or \Ctwo. In this sense, \GNFO and \Ctwo are convenient tools for explaining the decidability of other logics. Extensions of \GNFO have been studied that push the decidability frontier even further (for instance with fixed-point operators \cite{benedikt2016step, Benedikt2019:definability} and using clique-guards), but these fall outside the scope of this paper.

In an earlier line of investigation, Quine identified the decidable \emph{fluted fragment} (\FL) \cite{quine1969}, the first of several \emph{ordered logics} which have been the subject of recent interest \cite{purdy1996-1,purdy1996-2,purdy1999,purdy2002,pratt-hartmann2019}. The idea behind ordered logics is to restrict the order in which variables are allowed to occur in atomic formulas and quantifiers. Another recently introduced decidable fragment that falls in this family is the \emph{forward fragment} (\FF), whose syntax strictly generalizes that of \FL. Both \FL and \FF have the finite model property 
(FMP) \cite{pratt-hartmann2019} and embed \ML \cite{hustadt2004}, but are incomparable in expressive power to \GFO \cite{pratt-hartmann2016}, \FOtwo, and \UNFO. The recently-introduced \emph{adjacent fragment} (\AF) generalizes both \FF and \FL while also retaining decidability~\cite{bednarczyk2023limits}.

\begin{example}
The \FO-sentence $\exists x y (R(x,y) \land R(y,x))$ belongs to \GFO, \FOtwo and \UNFO, but is not expressible in \FF, since the structure consisting of two points with symmetric edges and the structure $(\mathbb{Z},S)$ with $S$ the successor relation, are ``infix bisimilar'' (see \cite{bednarczyk2022}).
\end{example}

\begin{figure}[t]\centering
\vspace{-2mm}
\newcommand{\yes}[1][\!\!\!\!\!]{\includegraphics[scale=.02]{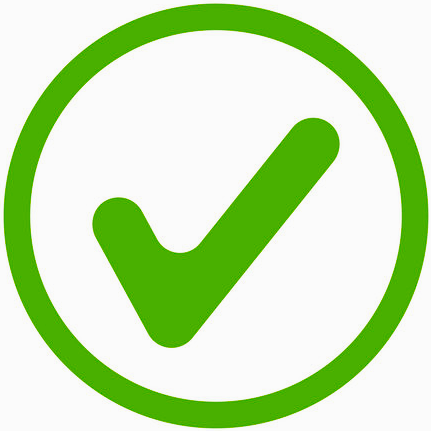}#1}
\newcommand{\no}[1][\!\!\!\!\!]{\includegraphics[scale=.02]{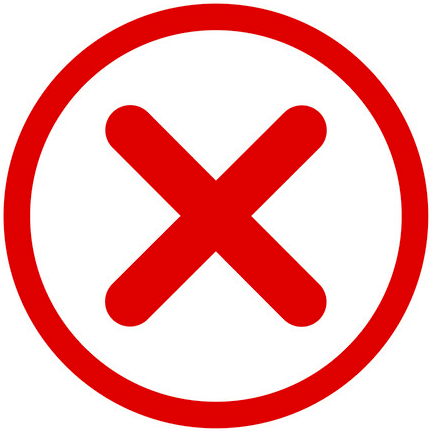}#1}
\begin{tikzpicture}
  [align=center,node distance=1.4cm, every node/.style={scale=1}]
    \node (fo) at (0,0) {\FO \yes};
    \node [below left of=fo] (ctwo) {\Ctwo \no};
    \node [below right of=fo] (gnfo) {\GNFO \yes};
    \node [below left  of=gnfo] (gfo)  {\GFO \no};
    \node [below right of=gnfo] (unfo) {\UNFO \yes};
    \node [below left of=ctwo] (twovar) {\FOtwo \no}; 
    \node [below right of= gfo] (ml) {Modal Logic \yes};
    \node [right of=gnfo] (ff) {FF \no}; 
    \node [right of=unfo] (fl) {FL \no}; 
    \draw [thick, shorten <=-2pt, shorten >=-2pt] (fo) -- (gnfo);
    \draw [thick, shorten <=-2pt, shorten >=-2pt] (fo) -- (ctwo);
    \draw [thick, shorten <=-2pt, shorten >=-2pt] (ctwo) -- (twovar);
    \draw [thick, shorten <=-2pt, shorten >=-2pt] (gnfo) -- node[below] {~~~(*)} (gfo);
    \draw [thick, shorten <=-2pt, shorten >=-2pt] (gnfo) -- (unfo);
    \draw [thick, shorten <=-2pt, shorten >=-2pt] (gfo) -- (ml);
    \draw [thick, shorten <=-2pt, shorten >=-2pt] (twovar) -- (ml);
    \draw [thick, shorten <=-2pt, shorten >=-2pt] (unfo) -- (ml);
      \draw [thick, shorten <=-2pt, shorten >=-2pt] (fo) -- (ff);
      \draw [thick, shorten <=-2pt, shorten >=-2pt] (ff) -- (fl);
      \draw [thick, shorten <=-2pt, shorten >=-2pt] (fl) -- (ml);
  
    \draw [dashed, thick] (-4,-1.5) -- 
      node[below, at start, sloped] {decidability} (2.5,0.1);
    \draw [dashed, thick] (-4,-2.4) -- node[below, at start, sloped] {finite model\\ property} (3,0.1);
\end{tikzpicture}

{
\scriptsize
\begin{tabular}[t]{l@{~~}l}
\FO & First-order logic\\
\FOtwo & Two-variable fragment \\
\Ctwo & Two-variable fragment  with counting \\
\GFO & Guarded fragment \\
\end{tabular} ~~
\begin{tabular}[t]{l@{~~}l}
\GNFO & Guarded-negation fragment \\
\UNFO & Unary-negation fragment \\
\FF & Forward fragment \\
\FL & Fluted fragment 
\end{tabular}
}

\caption{Landscape of decidable fragments of \FO with (\yes[]) and without (\no[]) CIP. The inclusion marked $(*)$ holds only for sentences and self-guarded formulas. Figured taken from~\cite{FOSSACS24}.}
\label{fig:fragments}

\end{figure}
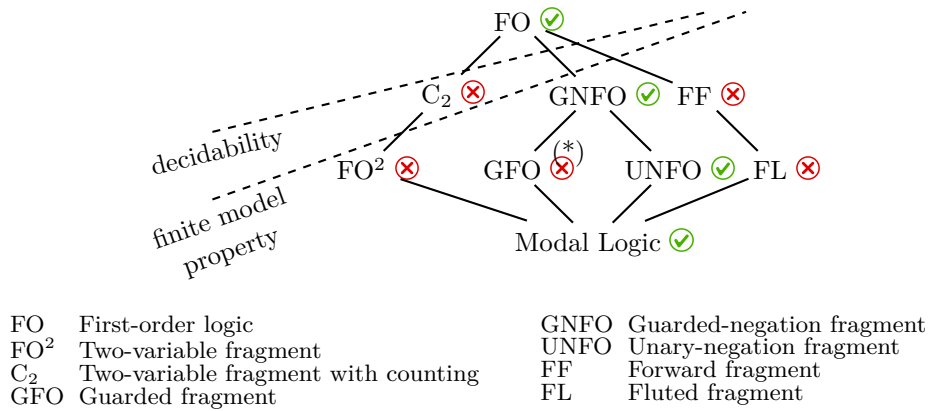

Ideally, an \FO-fragment is not only decidable, but also model-theoretically well behaved. A particularly important model-theoretic property of logics is the \emph{Craig Interpolation Property} (CIP). It turns out that, although \GFO is in many ways model-theoretically well-behaved~\cite{Andreka1998:Modal}, it lacks CIP~\cite{Hoogland02:interpolation}. Likewise, \FOtwo lacks CIP~\cite{Comer1969:classes} and the same holds for \Ctwo  (\cite[Example 2]{Jung2021:living} yields a counterexample). Both \FF and \FL lack CIP~\cite{bednarczyk2022}. On the other hand, \UNFO and \GNFO have CIP \cite{tencate2013:unary,Benedikt2013:rewriting}. Figure~\ref{fig:fragments} summarizes these known results. Note that we restrict attention to relational signatures without constant symbols and function symbols. Some of the results depend on this restriction. Other known results not reflected in Figure~\ref{fig:fragments} (to avoid clutter) are that the intersection of \GFO and \FOtwo (also known as \GFOtwo) has CIP~\cite{Hoogland02:interpolation}. Similarly, the intersection of \FF with \GFO and the intersection of \FL with \GFO (known as \GFF and \GFL, respectively) have CIP~\cite{bednarczyk2022}.

Additional results are known for \UNFO and \GNFO beyond the fact that they have CIP. In particular, CIP holds for the fixed-point extension of \UNFO, while the weak ``modal'' form of Craig interpolation, mentioned above, holds for the fixed-point extension of \GNFO~\cite{Benedikt2015:interpolation,Benedikt2019:definability}. Furthermore, interpolants for \UNFO and \GNFO can be constructed effectively, and tight bounds are known on the size of interpolants and the computational complexity of computing them~\cite{Benedikt2015:effective}.

\subsection{Repairing Interpolation}
When a logic $L$ lacks CIP, the question naturally arises as to whether there exists a more expressive logic $L'$ that has CIP. If such an $L'$ exists, then, in particular, interpolants for valid $L$-implications can be found in $L'$. This line of analysis is sometimes referred to as \emph{Repairing Interpolation} \cite{Areces03:repairing}. Another alternative is to find a logic $L'$ which contains interpolants for all valid implications in $L$ (see the $\textrm{Craig}(L,L')$ property in~\refchapter{chapter:modeltheory} and in~\refchapter{chapter:separation}.

If $L'$ is an \FO-fragment, and our aim is to repair interpolation by extension, then there is a trivial solution: \FO itself is an extension of $L$ satisfying CIP. More interesting is the following question: can a natural extension $L'$ of $L$ be identified which satisfies CIP while retaining decidability? In~\cite{FOSSACS24}, this question is answered in the negative: \FO is the smallest extension of \FOtwo and of \FF satisfying CIP. With some simple assumptions regarding the effective computability of the translation between the extension and the logic that it extends, it follows that no extension of \FOtwo or \FF satisfying CIP is decidable. It is also shown that the smallest logic extending \GFO that has CIP is \GNFO.

These results give us a clear sense of where, in the larger landscape of decidable fragments of \FO, we may find logics that enjoy CIP. What makes the above results remarkable is that, from the definition of the Craig interpolation property, it doesn't appear to follow that a logic without CIP would have a unique minimal extension with CIP. Note that a valid implication may have many possible interpolants, and the Craig interpolation property merely requires the existence of one such interpolant.  Nevertheless, the above results show that, in the case \FOtwo, \GFO, and \FF, such a unique minimal extension indeed exists (assuming suitable closure properties). It is worth noting that, in the case of these logics, their $\Delta$-closure (which is the smallest logic $\Delta(L)$ extending a logic $L$ and satisfying the $\Delta$-interpolation property) already satisfies the stronger Craig interpolation property, and hence is also a minimal extension satisfying Craig interpolation. For more on this, see \refchapter{chapter:modeltheory}.

The aforementioned line of work can be viewed as an instance of \emph{abstract model theory} for fragments of \FO. One large driving force behind the development of abstract model theory was the identification of \emph{extensions} of \FO which satisfy desirable model-theoretic properties, such as the compactness theorem, the L\"owenheim-Skolem, and Craig interpolation. One takeaway from this line of research is that CIP is scarce among many ``reasonable'' \FO-extensions. For an overview of results regarding interpolation for extensions of \FO, see~\refchapter{chapter:modeltheory}.

Some prior work exists that studies CIP for (extensions of) fragments of \FO using the lens of abstract model theory. It was shown in~\cite{tencate2005:interpolation} that the smallest extension of modal logic with the difference operator (\MLD) which satisfies CIP is full first-order logic. Additionally, in~\cite{GheerbrantCSL09}, the authors identified minimal extensions of various fragments of propositional linear temporal logic with CIP.  Repairing interpolation has also been pursued in the context of quantified modal logics, which typically lack CIP; in \cite{Areces03:repairing}, the authors showed that CIP can be repaired for such logics by adding nominals, @-operators and the $\downarrow$-binder.

\subsection{Alternative Approaches}

Several other approaches have been proposed for dealing with logics that lack CIP. One approach is to weaken CIP. For example, it was shown in \cite{Hoogland02:interpolation} that \GFO satisfies a weak, ``modal'' form of Craig interpolation, where, roughly speaking, only the relation symbols that occur in non-guard positions in the interpolant are required to occur both in the premise and the conclusion. As it turns out, this weakening of CIP is strong enough to entail the (non-projective) \emph{Beth Definability Property}, which is one important use case of CIP. We have discussed such weakenings of CIP in Sections~\ref{sec:delta} and~\ref{sec:beth}.

Another recent approach~\cite{Jung2021:living} is to develop algorithms for testing whether an interpolant exists for a given entailment. That is, rather than viewing Craig interpolation as a property of logics, the existence of interpolants is studied as an algorithmic problem at the level of individual entailments. The interpolant existence problem turns out to be indeed decidable (although of higher complexity than the satisfiability problem) for both \GFO and \FOtwo~\cite{Jung2021:living}. See~\refchapter{chapter:separation} for more details.

\section{Interpolation Relative to a Theory}
\label{sec:theories}

Example~\ref{ex:verification} involved Craig interpolation under an \FO theory. In this section
we briefly explore this in a bit more detail.

\begin{definition}[Craig interpolant under a theory]
    Let $\Sigma$ be an \FO theory
    and $\phi,\psi$ \FO formulas such that $\Sigma\models\phi\to\psi$.
    \begin{enumerate}
        \item
    A \emph{weak Craig interpolant for $\phi\to\psi$ under $\Sigma$} is a 
    \FO formula $\vartheta$ such that
        $\Sigma\models\phi\to\vartheta$, 
    $\Sigma\models\vartheta\to\psi$ and 
    $\sig(\vartheta)\subseteq (\sig(\phi)\cap \sig(\psi))\cup \sig(\Sigma)$ and $\fvar(\vartheta)\subseteq \fvar(\phi)\cap \fvar(\psi)$. 
    \item If, furthermore,
    $\sig(\vartheta)\subseteq \sig(\phi)\cap \sig(\psi)$, then we say that $\vartheta$ is a \emph{strong Craig interpolant for $\phi\to\psi$ under $\Sigma$}.
    \end{enumerate}
\end{definition}

Intuitively, weak Craig interpolants can contain non-logical symbols from the ambient theory $\Sigma$, even if they don't occur in $\phi$ or $\psi$, while strong Craig interpolants are strictly limited to using those symbols in both $\phi$ and $\psi$. As the next example shows, strong Craig interpolants do not always exists.

\begin{example}
    Let $\PA$ be the theory of Peano arithmetic (which is a theory of the natural numbers with addition, multiplication, and order). Let
    $\phi = \exists x(Px\land Qx\land \exists y(x=y+1))$
    and let $\psi = \exists x(Qx\land (x> 0 \lor R(x))$.
    Then 
    $\PA\models \phi\to\psi$.
    A weak Craig interpolant in this case is 
    $\exists x(Qx\land x> 0)$. Since $>$ is not a shared symbol of $\phi$ and $\psi$, this is not a strong Craig interpolant. 
    Indeed, it can easily be seen that there is no strong Craig interpolant in this case. Intuitively, the antecedent uses the symbols + and 1 to force that $x$ is positive, while the consequent uses the symbols $>$ and $0$ to express that $x$ is positive. A Craig interpolant can be found using either of these but cannot be constructed without using \emph{any} arithmetical symbols.
\lipicsEnd\end{example}

The next theorem shows that weak Craig interpolants always exist, and provides a sufficient condition for the existence of strong Craig interpolants. Let us say that an \FO theory $\Sigma$ is \emph{$(\sigma,\tau)$-splittable}
if it can be decomposed as $\Sigma=\Sigma_1\cup\Sigma_2$ with $(\sig(\Sigma_1)\cup\sigma)\cap (\sig(\Sigma_2)\cup\tau)\subseteq\sigma\cap\tau$.
For instance, it is not hard to see that
if each formula in $\Sigma$ involves only one symbol (i.e., if $|\sig(\phi)|=1$ for each 
$\phi\in\Sigma$), 
then $\Sigma$ is $(\sigma,\tau)$-splittable
for all choices of $\sigma$ and $\tau$.

\begin{theorem}
    Let $\Sigma$ be a \FO theory
    and $\phi,\psi$ \FO formulas such that $\Sigma\models\phi\to\psi$.
    \begin{enumerate}
        \item  There is a weak Craig interpolant for $\phi\to\psi$ under $\Sigma$.
    \item If $\Sigma$ is $(\sig(\phi),\sig(\psi))$-splittable, then there is a strong Craig interpolant for $\phi\to\psi$ under $\Sigma$.
    \end{enumerate}
\end{theorem}

\begin{proof}
    1. By compactness, we may assume without loss of generality that 
    $\Sigma$ is finite. If $\Sigma\models \phi\to\psi$, then
    $\models ((\bigwedge\Sigma)\land \phi)\to\psi$.
    Let $\vartheta$ be a Craig interpolant for this valid implication. Then, 
    clearly, $\Sigma\models\phi\to\vartheta$ and 
    $\Sigma\models\vartheta\to\psi$. Furthermore,
    $\sig(\vartheta)\subseteq (\sig(\phi)\cup \sig(\Sigma))\cap \sig(\psi)
    \subseteq (\sig(\phi)\cap \sig(\psi))\cup \sig(\Sigma)$.

    2. By compactness, we may assume without loss of generality that 
    $\Sigma$ is finite.
    If $\Sigma\models \phi\to\psi$, then
    $\models ((\bigwedge\Sigma_1)\land \phi)\to ((\bigwedge\Sigma_2)\to\psi)$,
    where $\Sigma_1 \cup \Sigma_2$ is a $(\sig(\phi),\sig(\psi))$-splitting of $\Sigma$.
    Let $\vartheta$ be a Craig interpolant for this valid implication. Then, 
    clearly, $\Sigma\models\phi\to\vartheta$ and 
    $\Sigma\models\vartheta\to\psi$. Furthermore,
    $\sig(\vartheta)\subseteq (\sig(\phi)\cup \sig(\Sigma_1))\cap (\sig(\psi)
    \cup\sig(\Sigma_2)) = \sig(\phi)\cap \sig(\psi)$.
\end{proof}

An interesting question is whether restricted fragments of \FO allow weak Craig interpolation under
certain first-order theories. We will give
two examples of results on this:
\begin{enumerate}
    \item The \emph{modal} fragment of \FO has weak Craig interpolation under every
    first-order theory consisting of universal Horn sentences (or more generally, relative to any first-order theory that is closed under bisimulation products). In particular, this implies that many relevant modal logics, such as 
    K4 and S4, have Craig interpolation.
    See~\refchapter{chapter:modal} for more details.
    \item The \emph{quantifier-free} fragment of \FO has (a restricted form of) weak Craig interpolation under every
    first-order theory that admits quantifier elimination. See~\refchapter{chapter:algebra} for more details.
\end{enumerate}

Other examples can be found in knowledge representation (see~\refchapter{chapter:kr}). For instance, the concept
language $\mathcal{ALC}$ has weak Craig interpolation  under $\mathcal{ALC}$-TBoxes \cite{tenCate2006:definitorially}.

Similarly, one may ask when \FO itself admits weak Craig interpolation relative to a \emph{second-order} theory. For instance, the class of finite structures is definable by a (finite) second-order
theory (cf.~Example~\ref{ex:so}), and 
Theorem~\ref{thm:failure-in-the-finite}
implies that \FO lacks weak Craig interpolation 
relative to this second-order theory. Similarly,
the class of $\omega$-words over a finite alphabet is definable by a (finite) second-order theory,
and \FO over this class of structures (or, equivalently by Kamp's theorem~\cite{Kamp1968}, \emph{Linear Temporal Logic}), lacks Craig interpolation~\cite{PlaceZeitoun2016,GheerbrantCSL09} (see~also~\refchapter{chapter:separation}).
\section{Proof Methods and Algorithmic Aspects}
\label{sec:techniques}

There are two main styles of techniques for proving
interpolation theorems for \FO, namely \emph{model theoretic} and \emph{proof theoretic}. 
At a high level, model theoretic proofs proceed along the following lines. Suppose, for the sake
of a contradiction, that some valid
implication $\models\phi\to\psi$ does \emph{not}
have a suitable interpolant. For simplicity,
assume that $\phi$ and $\psi$ are sentences. Using compactness,
it can then be shown that there exist structures
$A\models\phi$ and $B\models\neg\psi$ such that
$A$ and $B$ have the same \FO-theory in the 
signature $\sig(\phi)\cap\sig(\psi)$. Next, 
$A$ and $B$ are then ``amalgamated'' into 
a single structure (typically using $\omega$-saturated ultrapowers) that simultaneously satisfies
$\phi$ and $\neg\psi$, a contradiction. Note
that such proofs are inherently \emph{non-constructive}: they only show that interpolants exist but do not provide a way to construct an interpolant. \refchapter{chapter:modeltheory} contains such a model-theoretic proof of Craig interpolation for first-order logic.

On the other hand, \emph{proof theoretic} techniques construct an interpolant for $\models\phi\to\psi$ from a proof of $\phi\to\psi$ in a suitable proof system. Below, we show how proofs in a semantic tableaux can be used to construct interpolants. A detailed discussion on constructing interpolant from sequent-calculus proofs can be found in \refchapter{chapter:prooftheory} and \refchapter{chapter:fixedpoint}, and a discussion of more practical algorithms for computing interpolants can be found in \refchapter{chapter:automated} and \refchapter{chapter:verification}. We note that 
semantic tableaux, as a proof system for first-order logic, are closely related sequent calculi. Indeed,
a closed tableau for a formula, when read ``upside-down'', corresponds to a cut-free sequent calculus proof for its negation.

\subsection{A Proof of Craig interpolation via Semantic Tableaux}

We will now show how to prove a Craig interpolation theorem
via semantic tableaux. To keep the proof as simple as possible,
we will restrict attention to equality-free \FO sentences over signatures without function symbols: this restriction 
reduces the number of tableau rules, and allows us to avoid some technicalities regarding substitutability of variables by other variables. 
The proof we present, however, can be extended to the general case, 
and all the interpolation theorems for \FO that we discussed in Section~\ref{sec:fo} can be  
proved using the same approach. 
More detailed but accessible presentations of constructive proofs of  interpolation theorems for \FO via semantic tableaux can be found in the textbooks~\cite{fitting1996first} and~\cite{benedikt2016book}. 

\newcommand{\derrule}[2]{\begin{array}{c}#1\\\hline#2\end{array}}
\newcommand{\tabrule}[2]{\begin{array}{l}\bullet~ #1\\|\\\bullet~ #2\end{array}}
\newcommand{\labtabrule}[4]{\begin{array}{l}\parbox{3mm}{$#1$} ~\bullet~ #2\\~~~~|\\ \parbox{3mm}{$#3$} ~\bullet~ #4\end{array}}
\newcommand{\splitrule}[3]{\begin{array}{lll}&\bullet~ #1\\&\!\!\!\!\!/~~~~\backslash\\\bullet~ #2 && \bullet~ #3\end{array}}
\newcommand{\labsplitrule}[6]{\begin{array}{lll}&#1 ~\bullet~ #2\\&/~~~~\backslash\\#3~\bullet~ #4 && #5 ~\bullet~ #6\end{array}}
\newcommand{\clrule}[1]{\begin{array}{l}\bullet~ #1\\|\\\!\times\end{array}}
\newcommand{\labclrule}[2]{\begin{array}{l}\parbox{3mm}{$#1$} ~\bullet~ #2\\~~~~|\\~~~\times\end{array}}

A \emph{semantic tableau calculus} (or, simply, a \emph{tableau calculus}) is defined by a set of ``expansion rules'' and ``closure rules''. Typically, each rule pertains to a different connective.
A tableau is then a tree whose nodes are labeled by sentences, and where each non-leaf node corresponds to the application of an expansion rule, while each leaf node corresponds to the application of a closure rule. We will now present a concrete tableau calculus for equality-free \FO sentences over signatures without function symbols. We will assume that the input sentences
are in negation normal form.
The tableau calculus consists of four expansion rules and two closure rules. 
The expansion rules are: 

\[\tabrule{\phi_1\land\phi_2}{\phi_i ~ \text{for $i\in\{1,2\}$}} 
~~~~~~~~~ 
\tabrule{\exists x.\phi(x)}{\phi(c) ~ \text{for a fresh $c$}} ~ 
~~~~~~~~ 
    \begin{tikzcd}[column sep = 5mm, row sep = 3mm]
    & \bullet \rlap{~ $\phi\lor\psi$} \arrow[dash]{ld}\arrow[dash]{rd}\\    
    \bullet \rlap{~ $\phi$} && \bullet \rlap{~ $\psi$} 
    \end{tikzcd} 
~~~~~~~~~~~~~ 
\tabrule{\forall x.\phi(x)}{\phi(c)} \]
where, in the second and the fourth rule, $\phi(c)$ is obtained
from $\phi$ by replacing all free occurrences of $x$ by the constant symbol $c$.

Each rule should be read as follows: in order to fulfill the upper requirement, it is necessary to fulfill the lower requirements (or, in case of the third rule, one of the lower requirements). For instance, the first rule says that, in order to satisfy a conjunction, we must satisfy  both conjuncts. Similarly, the second rule says that,  to satisfy $\exists x\phi(x)$, we must satisfy $\phi(c)$  for some fresh constant $c$ (i.e., a constant symbol not yet occurring on the tableau branch).
The third rule is a ``branching'' rule. It says that, to satisfy a disjunction, we must satisfy one of the disjuncts. The fourth rule differs subtly from the second one: to satisfy $\forall x\phi(x)$ it is necessary to satisfy $\phi(c)$ for every constant $c$, not only for a fresh constant $c$. This rule, unlike the other three, may sometimes need to be
applied multiple times to the same formula (for different constants $c$). 

The closure rules are the following two rules:
\[\clrule{\bot} ~~~~~~~~~~ \clrule{\alpha,\neg\alpha}\]
Intuitively, these rules say that a branch may be closed if 
 either  $\bot$ occurs on the branch  or a formula and its negation both occur on the branch. The 
 $\times$ symbol is used to mark branches that are closed (i.e., branches ending in a leaf node).

If $\Sigma$ is a finite set of \FO sentences, then 
a \emph{closed tableau for $\Sigma$} is a finite tree whose root is labeled by $\Sigma$, where each non-leaf node corresponds to applying one of the 
expansion rules to a formula that has already been obtained (either in the previous step or at some earlier ancestor of the node in question),
and where each branch ends in the application of a closure rule. Note that, while the root node of the 
tableau is labeled by a finite set of formulas,
every other node is labeled by a single formula.

We can think of a tableau as a way to systematically search for a model that satisfies the input formulas. 
A closed tableau shows that no such model exists. Thus, a
closed tableau constitutes a proof of unsatisfiability. 

\begin{figure}[t]
\[\begin{tikzcd}[column sep = 15mm, row sep = 3mm]
& \bullet \rlap{~ $\begin{array}{l}{\color{black}\exists x ((A(x)\land \neg B(x))\land C(x))~^L}, \\ {\forall y((\neg A(y)\land E(y))\lor B(y))~^R}\end{array}$ \arrow[dash]{d}} \\
& \bullet \rlap{~ {\color{black}$ (A(c)\land \neg B(c))\land C(c) ~^L$}} \arrow[dash]{d}\\    
& \bullet \rlap{~ {\color{black}$ A(c) \land \neg B(c)~^L$}} \arrow[dash]{d}\\    
& \bullet \rlap{~ {\color{black}$ A(c)~^L$}} \arrow[dash]{d}\\    
& \bullet \rlap{~ {\color{black}$ \neg B(c)~^L$}} \arrow[dash]{d}\\    
& \bullet \rlap{~ {$ (\neg A(c)\land E(c))\lor B(c) ~^R$}} \arrow[dash]{ld}\arrow[dash]{rd}\\    
\bullet \rlap{~ {$ \neg A(c) \land E(c)~^R$}} \arrow[dash]{d} && \bullet \rlap{~ {$ B(c)~^R$}} \arrow[dash]{d}\\
\bullet \rlap{~ {$ \neg A(c), E(c)~^R$}} \arrow[dash]{d} && \!\times \\
\!\times
\end{tikzcd}\]
\caption{Example of a closed tableau.}\label{fig:tableau}
\end{figure}
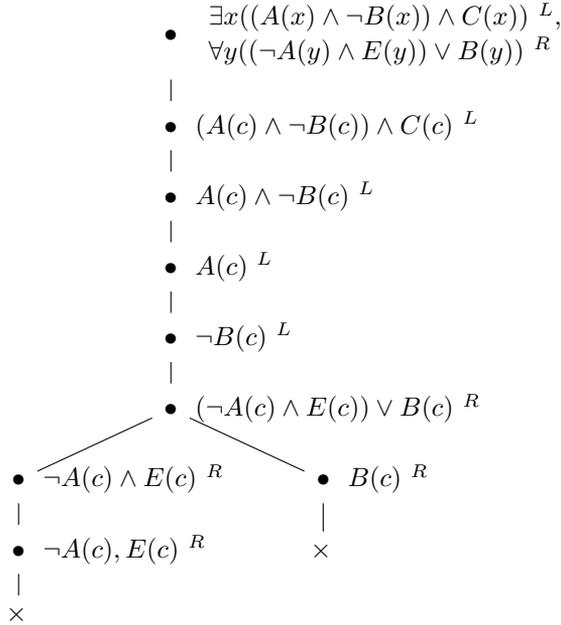

\begin{example} This example is adapted from {\cite[Example 2.2]{benedikt2016book}}.
    Consider the set of \FO sentences \[\Sigma=\{\exists x (A(x)\land \neg B(x)\land C(x)), \forall y((\neg A(y)\land E(y))\lor B(y))\}~.\] We show that $\Sigma$
    is unsatisfiable by giving, in Figure~\ref{fig:tableau}, a closed tableau for $\Sigma$. The annotation using ``L'' and ``R'' in this figure can be ignored for now, and will be explained later.
    Intuitively, the tableau describes a systematic search for a model of $\Sigma$. Starting with the set $\Sigma$, we repeatedly break down
     the requirements that need to be satisfied in order to make 
    $\Sigma$ true. When encountering a disjunction, we explore
    both possibilities. Since each possibility leads to a contradiction (the tableau is closed), we can conclude that it is not possible to satisfy
    $\Sigma$.
\lipicsEnd\end{example}

\begin{theorem}[Soundness and completeness~\cite{beth1955:semantic}]
\label{thm:tableau-completeness}
Let $\Sigma$ be a 
     set of equality-free \FO sentences $\Sigma$, in negation normal form, over a signature without function symbols. Then $\Sigma$ is 
    unsatisfiable if and only if there is a closed tableau for $\Sigma$.
\end{theorem}

We refer to~\cite{fitting1996first,benedikt2016book} for a proof of Theorem~\ref{thm:tableau-completeness}. In what follows,
we will show how to use the above theorem to
prove a Craig interpolation theorem. More precisely, we will prove
the following theorem. It is easy to see that this is an
equivalent rephrasing of the Craig interpolation theorem for \FO, 
restricted to equality-free sentences without function symbols.

\begin{theorem}
For all function-free, equality-free \FO sentences $\phi,\psi$, if 
$\{\phi,\psi\}$ is unsatisfiable then there is an function-free, equality-free
\FO sentence $\theta$ such that $\models\phi\to\theta$
and $\models\psi\to\neg\theta$ and $\sig(\theta)\subseteq\sig(\phi)\cap\sig(\psi)$.
\end{theorem}

\begin{proof}
Assume that $\{\phi,\psi\}$ is unsatisfiable. 
Then by  Theorem~\ref{thm:tableau-completeness}, there is a closed tableau for $\{\phi,\psi\}$.

First, we annotate each formula occurring on the tableau with L or R depending on whether
it was ``derived'' from $\phi$ or from $\psi$. Formally, we
label $\phi$ and $\psi$ as $L$ and $R$ respectively, and 
every formula which is derived on the basis of an expansion rule
gets the same label as the formula to which the expansion rule was applied. In this way, every formula occurring in the tableau is labeled with L or R. See Figure~\ref{fig:tableau} for an example.

Next, for every node $n$ of the tableau, let $\chi_L(n)$ be the conjunction of all L-labeled
sentences occurring at  node $n$ or at an ancestor of $n$, and let
$\chi_R(n)$ be the conjunction of all R-labeled
sentences occurring at  node $n$ or at an ancestor of $n$. 

\medskip\par\noindent\textbf{Claim:} For every node $n$ there is an equality-free FO sentence $\theta$ (which we will call the \emph{interpolant} for $n$) such that $\models\chi_L(n)\to\theta$ and
 $\models\chi_R(n)\to\neg\theta$ and  $\sig(\theta)\subseteq\sig(\chi_L(n))\cap\sig(\chi_R(n))$.

\medskip\par\noindent
This claim implies the theorem, as the interpolant of the root node is an interpolant for $\{\phi,\psi\}$ in the sense of the statement of the theorem.
We prove the claim by induction on the size of the subtree below $n$.
The base case is where  $n$ is a leaf, in which case one of the closure rules applies:
\[\labclrule{n}{\bot} ~~~~~~~~~~ \labclrule{n}{\alpha,\neg\alpha}\]
We can construct an interpolant for $n$ as follows:
\begin{itemize}
    \item If $\bot$ occurs on the branch as an L-labeled sentence, then $\bot$ is an interpolant.
    \item If $\bot$ occurs on the branch as an R-labeled sentence, then $\top$ is an interpolant.
\item If $\alpha$ occurs on the branch as an L-labeled sentence and $\neg\alpha$ occurs on the branch as an R-labeled sentence, then
    $\alpha$ is an interpolant.
    \item If $\alpha$ occurs on the branch as an R-labeled sentence and $\neg\alpha$ occurs on the branch as an L-labeled sentence, then
    $\neg\alpha$ is an interpolant.
    \item If $\alpha$ occurs on the branch as an L-labeled sentence and $\neg\alpha$ occurs on the branch as an L-labeled sentence, then
    $\bot$ is an interpolant.
    \item If $\alpha$ occurs on the branch as an R-labeled sentence and $\neg\alpha$ occurs on the branch as an R-labeled sentence, then
    $\top$ is an interpolant.
\end{itemize}

For non-leaf nodes $n$, we assume that an interpolant
has already been computed for each of the children nodes of $n$,
and we reason by cases depending on which of the four
expansion rules was applied at node $n$.

\[\labtabrule{n}{\alpha_1\land\alpha_2}{n'}{\alpha_i \text{~~[interpolant $\theta$]}}\]
Observe that, in this case, $\chi_L(n)$ and $\chi_L(n')$ are equivalent and have
the same signature, and the same holds for $\chi_R(n)$ and $\chi_R(n')$. It follows that $\theta$ is also 
an interpolant for $n$.

\[\labtabrule{n}{\exists x.\alpha(x)}{n'}{\alpha(c) ~ \text{for a fresh $c$}  \text{~~[interpolant $\theta$]}}\]
Again, in this case, it can be shown that $\theta$  is also an interpolant for $n$. Note that, since $c$ is fresh, it 
occurs only once, either in an L-labeled sentence or in an R-labeled sentence, and hence does not occur in the interpolant $\theta$.

\[    \begin{tikzcd}[column sep = 20mm, row sep = 10mm]
    & \llap{$n$ ~} \bullet \rlap{~ $\alpha\lor\beta$} \arrow[dash]{ld}\arrow[dash]{rd}\\    
    \llap{$n'_1$ ~} \bullet \rlap{~ $\alpha$ ~~[interpolant $\theta_1$]} && \llap{$n'_2$ ~} \bullet \rlap{~ $\beta$ ~~[interpolant $\theta_2$]} 
    \end{tikzcd} ~~~~~~~~~~~~~~~~~
\]
We now have two cases, depending on whether $\alpha\lor\beta$ occurs on the branch as an $L$-labeled sentence
or an $R$-labeled sentence. In the former case, $\theta_1\lor\theta_2$ is an interpolant for $n$. Otherwise, $\theta_1\land\theta_2$ is an interpolant for $n$.

\[\labtabrule{n}{\forall x.\alpha(x)}{n'}{\alpha(c)  \text{~~[interpolant $\theta$]}} \]
We have four cases, depending on whether $\forall x.\alpha(x)$ occurs on the branch as an $L$-labeled sentence
or an $R$-labeled sentence, and also on whether $c$ occurs on the branch in an $L$-labeled sentence or an $R$-labeled sentence.
\begin{itemize}
    \item If $\forall x.\alpha(x)$ occurs on the branch as an $L$-labeled sentence and $c$ occurs on the branch in an $L$-labeled sentence, then $\theta$ is an interpolant for $n$. 
    \item If $\forall x.\alpha(x)$ occurs on the branch as an $R$-labeled sentence and $c$ occurs on the branch in an $R$-labeled sentence, then $\theta$ is an interpolant for $n$. 
    \item If $\forall x.\alpha(x)$ occurs on the branch as an $L$-labeled sentence and $c$ does \emph{not} occur on the branch in an $L$-labeled sentence, then $\forall x.\theta[c/x]$ is an interpolant for $n$. 
    \item If $\forall x.\alpha(x)$ occurs on the branch as an $R$-labeled sentence and $c$ does \emph{not} occur on the branch in an $R$-labeled sentence, then $\exists x.\theta[c/x]$ is an interpolant for $n$.  
\end{itemize}
\end{proof}

\begin{figure}
\textbf{Expansion rules:}
\[
\begin{tikzcd}[column sep = 20mm, row sep = 10mm]
    \bullet \rlap{~ $\alpha_1\land\alpha_2$ ~~ {\color{darkgray}[interpolant $\theta$]}} 
    \arrow[dash]{d}
    \arrow[dashed,gray,transform canvas={xshift=15ex},<-]{d}
    \\    
    \bullet \rlap{~ $\alpha_i$ ~~~~~~~~~~{\color{darkgray}[interpolant $\theta$]}} 
\end{tikzcd}~~~~~~~~~~~~~~~~~~~~~~~~~~~~~~~~~~~~~~~~~~~~~~~~~~~~
\begin{tikzcd}[column sep = 20mm, row sep = 10mm]
    \bullet \rlap{~ $\exists x \alpha(x)$ ~~~~~~~~~~~~ {\color{darkgray}[interpolant $\theta$]}} 
    \arrow[dash]{d}
    \arrow[dashed,gray,transform canvas={xshift=25ex},<-]{d}
    \\    
    \bullet \rlap{~ $\alpha(c)$ for fresh $c$ ~~{\color{darkgray}[interpolant $\theta$]}} 
\end{tikzcd}~~~~~~~~~~~~~`
\]
\[    
\begin{tikzcd}[column sep = 20mm, row sep = 10mm]
    & \bullet \rlap{~ $\alpha\lor\beta$ ~~ {\color{darkgray}[interpolant $\theta$]}} 
    \arrow[dash]{ld}
    \arrow[dashed,gray,transform canvas={xshift=16ex},<-]{ld}
    \arrow[dash,crossing over]{rd}
    \arrow[dashed,gray,transform canvas={xshift=16ex},<-]{rd}
    \\    
    \bullet \rlap{~ $\alpha$ ~~{\color{darkgray}[interpolant $\theta_1$]}} && 
    \bullet \rlap{~ $\beta$ ~~{\color{darkgray}[interpolant $\theta_2$]}} 
\end{tikzcd}~~~~~~~~~~~~~`
\]
\begin{center}
where {\color{darkgray}$\theta = \begin{cases} 
\theta_1\lor\theta_2 & \text{if $\alpha\lor\beta$ is $L$-labeled}\\
\theta_1\land\theta_2 & \text{if $\alpha\lor\beta$ is  $R$-labeled}
\end{cases}$}
\end{center}
\[
\begin{tikzcd}[column sep = 20mm, row sep = 10mm]
    \bullet \rlap{~ $\forall x \alpha(x)$ ~~ {\color{darkgray}[interpolant $\theta'$]}} 
    \arrow[dash]{d}
    \arrow[dashed,gray,transform canvas={xshift=15ex},<-]{d}
    \\    
    \bullet \rlap{~ $\alpha(c)$ ~~~~~~~{\color{darkgray}[interpolant $\theta$]}} 
\end{tikzcd}~~~~~~~~~~~~~~~~~~~~~~~~~~~~~~~~~~~~~~~~
\text{where {\color{darkgray}$\theta' = \begin{cases} 
\exists x\theta[c/x] & \text{if $\forall x\alpha(x)$ is $L$-labeled and $c$ only} \\& \text{occurs in $L$-labeled formulas}\\
\forall x\theta[c/x] & \text{if $\forall x\alpha(x)$ is $R$-labeled and $c$ only} \\& \text{occurs in $R$-labeled formulas }\\
\theta & \text{otherwise}\\
\end{cases}$}}
\]

\bigskip

\textbf{Closure rules:}
\[
\begin{tikzcd}[column sep = 20mm, row sep = 10mm]
    \bullet \rlap{~ $\alpha, \neg\alpha$ ~~ {\color{darkgray}interpolant $\theta$}}
    \arrow[dash]{d}
    \\    
    \times
\end{tikzcd}
~~~~~~~~~~~~~~~~~~~~~~~~~~~~~~~~~~~~
\text{where {\color{darkgray}$\theta = \begin{cases} 
\top & \text{if $\alpha$ and $\neg\alpha$ are L-labeled}\\
\bot & \text{if $\alpha$ and $\neg\alpha$ are R-labeled}\\
\alpha & \text{if $\alpha$ is $L$-labeled and $\neg\alpha$ is R-labeled}\\
\neg\alpha & \text{if $\alpha$ is $R$-labeled and $\neg\alpha$ is L-labeled}
\end{cases}$}}
\]
\bigskip
\[
\begin{tikzcd}[column sep = 20mm, row sep = 10mm]
    \bullet \rlap{~ $\bot$ ~~ {\color{darkgray}interpolant $\theta$}}
    \arrow[dash]{d}
    \\    
    \times
\end{tikzcd}
~~~~~~~~~~~~~~~~~~~~~~~~~~~~~~~~~~~~
\text{where {\color{darkgray}$\theta = \begin{cases} 
\bot & \text{if $\bot$ is L-labeled}\\
\top & \text{if $\bot$ is R-labeled}\end{cases}$}}
~~~~~~~~~~~~~~~~~~~~~~~~~~~~~
\]

\caption{Tableau rules with corresponding interpolant propagation rules}
\label{fig:propagation-rules}
\end{figure}

The interpolant constructions in the above proof are summarized in Figure~\ref{fig:propagation-rules}.

\subsection{Computational Complexity and Size Bounds}
\label{sec:algorithmic}

As described above, proof-theoretic proofs of Craig interpolation construct an interpolant for $\phi\to\psi$ from a proof of the validity of $\phi\to\psi$ (in a suitable proof system such as semantic tableaux). It follows from the undecidability of \FO that  there is no recursive function for constructing a proof of a valid \FO implication, so such proof-theoretic approaches are \emph{constructive} only in a weak sense: the interpolant can be effectively constructed once we have laid our hands on a suitable proof of $\phi\to\psi$. This is the best one
can hope for:

\begin{theorem}[\cite{Friedman1976}]
Let $|\phi|$ denote the length of an \FO formula $\phi$.
There is no recursive function $f:\mathbb{N}\to\mathbb{N}$ such that
every valid \FO implication 
$\models\phi\to\psi$ has a Craig interpolant $\vartheta$
with $|\vartheta|\leq f(|\phi|+|\psi|)$.
\end{theorem}

In particular, this implies that there is no recursive function taking as input two \FO formulas $\phi,\psi$ and returning a Craig 
interpolant whenever $\models\phi\to\psi$.

As mentioned in Section \ref{sec:frag}, for decidable fragments of \FO, the situation improves regarding the size and the effective computability of interpolants. More precisely, for present purposes, let a \emph{decidable fragment} $F$ of \FO be a fragment that is decidable for entailment (i.e., it is decidable whether $\models\phi\to\psi$ for $\phi,\psi\in F$) and such that it is decidable whether a given \FO formula belongs to $F$.

\begin{proposition}
    Let $F$ be any decidable fragment of \FO.
    \begin{itemize}
        \item
    There is a recursive function that takes as input
    $\phi,\psi\in F$ and produces a Craig interpolant in \FO whenever $\models\phi\to\psi$.
   
    \item If $F$ has CIP, then there is a recursive function that takes as input
    $\phi,\psi\in F$ and produces a Craig interpolant  $\vartheta\in F$ whenever $\models\phi\to\psi$.
    \end{itemize}
\end{proposition}

\begin{proof}
    We first test the validity of the entailment $\models\phi\to\psi$. If it is not valid, we terminate, otherwise we search for an interpolant by brute force enumeration, enumerating all triples of strings $(\vartheta,p_1,p_2)$ where $\vartheta$ is a formula (in \FO respectively in $F$) 
    with $\sig(\vartheta)\subseteq\sig(\phi)\cap\sig(\psi)$ and $\fvar(\vartheta)\subseteq\fvar(\phi)\cap\fvar(\psi)$ and where 
    $p_1,p_2$ are valid proofs
     (in some sound and complete proof system for \FO) for $\phi\to\vartheta$ and $\vartheta\to\psi$. This is guaranteed to terminate, since we know that a Craig interpolant exists.
\end{proof}

The above proposition does not provide any concrete bounds on the size or complexity of computing interpolants. For various fragments, however, concrete bounds have been obtained. See for instance~\refchapter{chapter:modal} for the case of modal logic and~\refchapter{chapter:verification} for the quantifier-free fragment of \FO. Similarly for \GNFO, tight bounds are known for the size of Craig interpolants as well as an optimal algorithm for computing them~\cite{Benedikt2015:effective}, using the method of type elimination sequences (see~\refchapter{chapter:modal}).

We conclude with an open question. Recall that \FO formulas do not in general have uniform interpolants. 

\begin{question}
    Consider the following problem: \emph{given an \FO-formula $\phi$
    and a signature $\tau\subseteq\sig(\phi)$, does $\phi$ have
    a uniform interpolant w.r.t.~$\tau$?} This problem is easily seen to be undecidable and to be in $\Sigma^0_3$ (the third level of the arithmetical hierarchy).
    Is it $\Sigma^0_3$-complete?
\end{question}

\bibliography{admin/bib.bib}

\end{document}